\documentclass[12pt,draftclsnofoot,onecolumn,a4paper]{IEEEtran}

\usepackage{longtable}
\usepackage{graphicx}
\usepackage{subfigure}
\usepackage{epstopdf}
\usepackage{citesort}
\usepackage{bm}
\usepackage{amssymb}
\usepackage{amsthm}
\usepackage{amsfonts}
\usepackage{amsmath}
\usepackage{epsfig}
\usepackage{fancybox}
\usepackage{textcomp}
\usepackage{multirow}
\usepackage{booktabs}
\usepackage{mathrsfs}
\usepackage{pifont}
\usepackage{bbm}
\usepackage[ruled]{algorithm2e}
\usepackage{xcolor}
\usepackage{euscript}

\newtheorem{theorem}{Theorem}
\newtheorem{lemma}{Lemma}
\newtheorem{corollary}{Corollary}

\newlength{\halfpagewidth}
\divide\halfpagewidth by 2

\makeatletter
\def\ScaleIfNeeded{%
\ifdim\Gin@nat@width>\linewidth \linewidth \else \Gin@nat@width
\fi } \makeatother

\begin{document}
%

\title{Wireless Power Transfer in Massive MIMO Aided HetNets with User Association}

\author{Yongxu Zhu, Lifeng Wang~\IEEEmembership{Member,~IEEE,} Kai-Kit Wong,~\IEEEmembership{Fellow,~IEEE,} \\  Shi Jin~\IEEEmembership{Member,~IEEE,}
and Zhongbin Zheng
\thanks{Y. Zhu, L. Wang, and K.-K. Wong are with the Department of Electronic and Electrical Engineering, University College London (UCL), London, UK(Email: $\rm\{yongxu.zhu.13,lifeng.wang, kai$-$\rm kit.wong\}@ucl.ac.uk$).}
\thanks{S. Jin is with National Mobile Communications Research Laboratory, Southeast University, Nanjing 210096, China (Email: $\rm{jinshi}@seu.edu.cn$).}
\thanks{Z. Zheng is with China Academy of Information and Communications Technology, China (Email: $\rm{ben}@ecit.org.cn$).
}
}

\maketitle

\begin{abstract}
This paper explores the potential of wireless power transfer (WPT) in massive multiple-input multiple-output (MIMO) aided heterogeneous networks (HetNets), where massive MIMO is applied in the macrocells, and users aim to harvest as much energy as possible and reduce the uplink path loss for enhancing their information transfer. By addressing the impact of massive MIMO on the user association, we compare and analyze user association schemes: 1) downlink received signal power (DRSP) based approach for maximizing the harvested energy; and 2) uplink received signal power (URSP) based approach for minimizing the uplink path loss.  We adopt the linear maximal-ratio transmission (MRT) beamforming for massive MIMO power transfer to recharge users. By deriving new statistical properties, we obtain the exact and asymptotic expressions for the average harvested energy. Then we derive the average uplink achievable rate under the harvested energy constraint. Numerical results demonstrate that the use of massive MIMO antennas can improve both the users' harvested energy and uplink achievable rate in the HetNets, however, it has negligible effect on the ambient RF energy harvesting.  Serving more users in the massive MIMO macrocells will  deteriorate the uplink information transfer because of less harvested energy and more uplink interference. Moreover, although DRSP-based user association harvests more energy to provide larger uplink transmit power than the URSP-based one in the massive MIMO HetNets,  URSP-based user association could achieve better performance in the uplink information transmission.
\end{abstract}
\vspace{-1 cm}


\section{Introduction}
Traditional energy harvesting sources such as solar, wind, and hydroelectric power highly depend upon time and locations, as well as the conditions of the environments. Wireless power transfer (WPT) in contrast is a much more controllable approach to prolong the lifetime of mobile devices~\cite{Rui_Zhang_2013,Kaibin2014,CheDZ14}. Additionally, the potentially harmful interference received by the energy harvester can actually become a useful energy source. Recently, the potential of harvesting the ambient energy in the fifth-generation (5G) networks has been studied in~\cite{Hossain_WC_Mag_2014,Yi_Liu_2015,dantong-survey}.

Heterogeneous networks (HetNets) are identified as one of the key enablers for 5G, e.g., \cite{Hossain_WC_Mag_2014,AJG14}. In HetNets, small cells are densely deployed~\cite{AJG14,HE15}, which shortens the distances between the mobile devices and the base stations (BSs). {Recently, there is an interesting integration between WPT and HetNets, suggesting that stations, referred to as power beacons (PBs), can be deployed in cellular networks for powering users via WPT \cite{Kaibin2014}. In \cite{Tabassum_2015} and \cite{S_Bi_2015},  the optimal placement of power beacons in the cellular networks has been investigated.
}

Recent attempts have been to understand the feasibility of WPT in cellular networks, device-to-device (D2D) communications and sensor networks. In particular, both picocell BSs and energy towers (or PBs) were considered in \cite{Erol-Kantarci2014} to transfer energy to the users, and their problem was to jointly maximize the received energy and minimize the number of active picocell BSs and PBs. Subsequently in \cite{H_T_WPT_2015}, user selection policies in dedicated RF-powered uplink cellular networks were investigated, where the BSs acted as dedicated power sources. Further, \cite{S_A_H_WPT_2015} studied a $K$-tier uplink cellular network with energy harvesting, where the cellular users harvested the RF energy from the concurrent downlink transmissions in all network tiers. Then \cite{Sakr_AH_2015} studied the D2D scenario in which the cognitive transmitters harvested energy from the interference to support the communication. As mentioned in  \cite{Kaibin_Huang2015}, however, ambient RF energy harvesting is sufficient only for powering low-power sensors with sporadic activities, and dedicated energy source is required for powering mobile devices such as smartphones. As such, \cite{yuanwei_ICC2015} turned the attention to the case, where D2D transmitters harvested energy from the PBs, and proposed several power transfer policies. {In \cite{Xiao_Lu_2015}, battery-free sensor node harvested energy from the access point and ambient RF transmitters based on the power splitting architecture, and the locations of RF transmitters were modeled using Ginibre $\alpha$-DPP.}

On the other hand, {massive multiple-input multiple-output (MIMO) systems, using a large number of antennas at the BSs, achieve ultra-high spectral efficiency by accommodating a large number of users in the same radio channel \cite{Marzetta2010}. For massive MIMO to become reality, there are still some issues such as  high circuit power consumption~\cite{AJG14}, which need to be addressed.}  The exceptional spatial selectivity of massive MIMO means that very sharp signal beams can be formed \cite{ngo2013energy,Lifeng_massiveMIMO} and of great importance to WPT. {Motivated by this, \cite{Xiaoming_chen2013} studied the wireless information and power transfer in a point-to-point (P2P) system including a single-antenna user and its serving BS equipped with large antenna array, where energy efficiency for uplink information transfer was maximized under the quality-of-service (QoS) constraint. {Later in \cite{HengzhiWang_2014}, a receiver with large number of antennas was assumed to harvest energy from a single-antenna transmitter and a single-antenna interferer, and an algorithm was proposed to maximize the data rate while guaranteeing a minimum harvested energy with a large receive antenna array using antenna partitioning.} {In contrast to \cite{Xiaoming_chen2013,HengzhiWang_2014},  \cite{Gang_yang2015} considered the uplink throughput optimization in a single massive MIMO powered cell, where  an access point equipped with a large antenna array transfers energy to multiple users.} {The opportunities and challenges of deploying a massive number of distributed antennas for WPT was discussed in~\cite{Kai_kit_Wong_2015_Mag}.} { In addition, the shorter wavelengths at the mmWave frequencies enable mmWave BSs to pack more antennas for achieving large array gains. Hence recent research works such as~\cite{Robert_heath_mmWave_energy,Lifeng_wang2015_globcom} also studied WPT in mmWave cellular networks. Particularly, in \cite{Robert_heath_mmWave_energy},   the mmWave antenna beam was characterized by using the sectored antenna model and the energy coverage probability was evaluated.  In \cite{Lifeng_wang2015_globcom}, uniform linear array (ULA) with analog beamforming was considered for WPT in mmWave cellular networks.  Different from~\cite{Robert_heath_mmWave_energy,Lifeng_wang2015_globcom}, this paper  focuses on massive MIMO enabled wireless power transfer with digital beamforming in the conventional cellular bands, which will be detailed later.}

Regarded as a promising network architecture to meet the increasing demand for mobile data, massive MIMO empowered HetNets have recently attracted much attention~\cite{BE2013,DBLP2014,NW,XuM15,D_Liu_2015}. {In \cite{BE2013}, downlink beamforming design for minimizing the power consumption was investigated in a single massive MIMO enabled macrocell overlaid with multiple small cells, and it was shown that total power cost can be significantly reduced while satisfying the QoS constraints. Motivated by these research efforts, in this paper, we explore the potential benefits of massive MIMO  HetNets for  wireless information and power transfer (WPT and wireless information transfer (WIT)), which is novel and has not been conducted yet.}

{Different from the aforementioned literature such as \cite{Xiaoming_chen2013,HengzhiWang_2014,Gang_yang2015} where WPT and WIT were only considered in a single cell, we study massive MIMO antennas being harnessed in the macrocells, and employ a stochastic geometry approach to model the $K$-tier HetNets. In particular,  users first harvest energy from downlink WPT, and then use the harvested energy for WIT in the uplink.} {In this scenario, user association determines whether a user is associated with a particular base station for downlink WPT in such networks, and therefore it is crucial to study the effect of user association on WPT. The work of \cite{S_A_H_WPT_2015} considered that users relied on ambient RF energy harvesting, and only studied the effect of user association on uplink information transmission.} {User association in massive MIMO HetNets has been recently investigated for optimizing the throughput~\cite{DBLP2014,NW,XuM15} and energy efficiency~\cite{D_Liu_2015}. The effect of using different user association methods on WPT in such networks is unknown. Hence we examine the effect of user association on the WPT and WIT in massive MIMO HetNets by considering two user association methods: (1) downlink received signal power (DRSP) based for maximum harvested energy, and (2) uplink received signal power (URSP) based for minimum uplink path loss.}  One of our aims is to find out which scheme is better for uplink WIT. In this paper, we have made the following contributions:
\begin{itemize}
\item We develop an analytical framework to examine the implementation of downlink WPT and uplink WIT in massive MIMO aided HetNets with stochastic geometric model. As the intra-tier interference is the source of energy, interference avoidance is not required and maximal-ratio transmission (MRT) beamforming is used for WPT for multiple users in the macrocells.
\item {We investigate the impacts of massive MIMO on the user association of the HetNets, and examine both DRSP-based and URSP-based algorithms by deriving the exact and asymptotic expressions for the probability of a user associated with a macrocell or a small cell in the HetNet.}
\item {We derive the exact and asymptotic expressions for the average harvested energy when users are equipped with large energy storage. We show that the asymptotic expressions can well approximate the exact ones. The implementation of massive MIMO can significantly increase the harvested energy in the HetNets, since it provides larger power gain for users served in the macrocells,  and enables that users with higher received power are offloaded to the small cells.\footnote{Note that power gain is also referred to as array gain in the literature.} In addition, DRSP-based user association scheme outperforms URSP-based in terms of harvested energy, which means that it supports higher user transmit power for uplink information transmission.}
\item {We derive the average uplink achievable rate supported by the harvested energy.}
{Our results demonstrate that the uplink performance is enhanced by increasing the number of antennas at
the macrocell BS, but {serving more users in the macrocells decreases the average achievable rate
because of lower uplink transmit power and more severe uplink interference.} For the case of dense small cells, it can still be interference-limited in the uplink. Furthermore, {although DRSP-based
user association scheme harvests more energy to provide larger uplink transmit power, URSP-based can achieve better WIT performance in the uplink.}} 
\end{itemize}

{The notation of this paper is shown in Table~\ref{tab:gain}.}

\begin{table}\footnotesize
\centering
\caption{Notation}\label{tab:gain}
\begin{tabular}{c|c}
\toprule[1pt]
 \\
$\Phi_\text{M}$, $\lambda_\text{M}$  & Macrocells PPP and density   \\
 \hline
$\Phi_i$, $\lambda_i$  & $i$-th tier PPP and density \\
 \hline
 $T$, $\tau$  &  One block time and time allocation factor \\
 \hline
$N$  & Number of antennas   \\
 \hline
$S$  & Number of single-antenna users served by a MBS\\
 \hline
$P_\text{M}$, $P_i$ & MBS and $i$-th tier transmit power \\
 \hline
$\alpha_\text{M}$, $\alpha_i$ & MBS and $i$-th tier pass loss exponent  \\
 \hline
$G^\text{D}_a$, $G^\text{U}_a$ & Downlink and uplink power gain \\
 \hline
$d$ & Reference distance \\
 \hline
$h$, $g$&  Small-scale fading channel power gain\\
 \hline
 $\Gamma\left(\vartheta,\theta\right)$ & Gamma distribution with shape $\vartheta$ and scale $\theta$\\
  \hline
 $\exp\left(z\right)$ &   Exponential distribution with the parameter $z$\\
 \hline
 $\tilde{\mathcal{U}}_\mathrm{M}$,  $\tilde{\mathcal{U}}_i$  & Interfering users PPP in the MBS tier and the $i$-th tier
\\
 \hline
$P_{u_o}$ & Typical user's transmit power \\
 \hline
${\rm{\mathbf{1}}}\left(\cdot\right)$ & Indicator function \\
 \hline
 $\mathbb{E}\left\{\cdot\right\}$ & Expectation operator\\
\toprule[1pt]
\end{tabular}
\end{table}

\section{Network Description}
This paper considers a $K$-tier time-division duplex (TDD) HetNet including macrocells and small cells such as picocells and relays, etc. {Each user first harvests the energy from its serving BS (as a dedicated RF energy source) in the downlink,} and uses the harvested energy for WIT in the uplink. Let $T$ be the duration of a communication block. The first and second sub-blocks of duration $\tau T$ and $\left(1-\tau\right)T$ are allocated to the downlink WPT and uplink WIT, respectively, where $\tau \left(0\leq \tau \leq 1\right)$ is the time allocation factor. We assume that the first tier represents the class of macrocell BSs (MBSs), each of which is equipped with a large antenna array~\cite{Jungnickel_IEEE_Commag}. The locations of the MBSs are modelled using a homogeneous Poisson point process (HPPP) $\Phi_\mathrm{M}$ with density $\lambda_\mathrm{M}$.
{The locations of the small-cell (such as micro/picocell, femtocell, etc.)  BSs (SBSs) in the $i$-th tier ($i=2,\dots,K$)
are modelled by an independent HPPP $\Phi_i$ with density $\lambda_i$.}
{It is assumed that the density of users is much greater than that of BSs so that there always will be one active mobile user at each time slot in every small cell and hence multiple active mobile users in every macrocell.\footnote{In reality, there may be more than one active users in a small cell and this can be dealt with using multiple access techniques.}} In the macrocell, $S$ single-antenna users communicate with an $N$-antenna MBS (assuming $N \gg S \geq 1$) in the uplink over the same time slot and frequency band.{\footnote{{We note that in~\cite{Sakr_AH_2015}, the probability mass function of the number of users served by a generic BS was derived by approximating the area of a Voronoi cell via a gamma-distributed random variable. However, the result in~\cite{Sakr_AH_2015} cannot be applied in this paper,  since the Euclidean plane is not divided into Voronoi cells based on the considered user association methods. We highlight that it is an important work to study the case of the dynamic $S$ following a certain distribution in less-dense scenarios.} }}  In the small cell, only one single-antenna user is allowed to communicate with a single-antenna SBS at a time slot. We assume that perfect channel state information (CSI) is known at the BS,\footnote{In the practical TDD massive MIMO systems, the downlink CSI can be obtained through channel reciprocity based on uplink training~\cite{Marzetta_2010_Nonc}.} and {the effect of pilot contamination on channel estimation is omitted. As mentioned in~\cite{AJG14,Hoon2012},  pilot contamination is a relatively
secondary factor for all but colossal numbers of antennas, and various methods to mitigate
pilot contamination via low-intensity base station coordination have already been
proposed in the literature such as \cite{YinHai_fan_2013}.}  In addition, universal frequency reuse is employed such that all of the tiers share the same bandwidth and all the channels are assumed to undergo independent identically distributed (i.i.d.) quasi-static Rayleigh block fading.

\subsection{User Association}
{We introduce two user association algorithms: (1) a user is associated with the BS based on the maximum DRSP at the user, which results in the largest average received power; and (2) a user is associated with the BS based on the maximum URSP at the BS, which will minimize the power loss of user's signal during the propagation.}\footnote{{Although user association for the downlink and uplink can be decoupled to maximize both the DRSP and URSP, the main drawback for the decoupled access is that channel reciprocity in massive MIMO systems will be lost~\cite{Federico_Boccardi2015}.}}

Considering the effect of massive MIMO, the average received power at a user that is connected with the $\ell$-th MBS ($\ell \in \Phi_\mathrm{M}$) can be expressed as
\begin{align}\label{Macro_Receive_Power}
{P_{r,\ell}} =G_a^{\mathrm{D}} \frac{P_\mathrm{M}}{S}L\left(\left|X_{\ell,\mathrm{M}}\right|\right),
\end{align}
where $G_a^{\mathrm{D}}$ denotes the power gain obtained by the user associated with the MBS, $P_\mathrm{M}$ is the MBS's transmit power, $L\left(\left|X_{\ell,\mathrm{M}}\right|\right)=\beta{ {{\left|X_{\ell,\mathrm{M}}\right|}}^{ - {\alpha_\mathrm{M}}}}$ is
the path loss function, $\beta$ is the frequency dependent constant value, $\left|X_{\ell,\mathrm{M}}\right|$ denotes the distance, and $\alpha_\mathrm{M}$ is the path loss exponent. In the small cell, the average received power at a user that is connected with the $j$-th SBS ($j \in \Phi_i$) in the $i$-th tier is expressed as
\begin{align}\label{Small_Receive_Power}
{P_{r,i}} ={P_i}L\left(\left|X_{j,i}\right|\right),
\end{align}
where ${P_i}$ denotes the SBS's transmit power in the $i$-th tier and as above $L\left(\left|X_{j,i}\right|\right)=\beta{\left( {{\left|X_{j,i}\right|}} \right)^{ - {\alpha_i}}}$ is the path loss function with distance $\left|X_{j,i}\right|$ and path loss exponent $\alpha_i$.

For DRSP-based user association, the aim is to maximize the average received power. Thus, the serving BS for a typical user is selected according to the following criterion:
\begin{align}\label{DRSP_UA}
\mathrm{BS}: \arg \max_{k \in \left\{ {\mathrm{M},2, \dots ,K} \right\}} P^*_{r,k},
\end{align}
where
\begin{equation}
P^*_{r,\mathrm{M}}= \mathop{ \max} \limits_{\ell\in \Phi_\mathrm{M}} P_{r,\ell},~\mbox{and}~
P^*_{r,i}= \mathop{ \max} \limits_{j\in \Phi_i} P_{r,i}.
\end{equation}

By contrast, for URSP-based user association, the objective is to minimize the uplink path loss, and as such, the serving BS for a typical user is selected by
\begin{align}\label{URSP_UA}
\mathrm{BS}:  \arg \max_{k \in \left\{ {\mathrm{M},2, \ldots ,K} \right\}} L^*\left(\left| X_k\right| \right),
\end{align}
where
\begin{align}
L^*\left(\left|X_\mathrm{M}\right|\right)&= G_a^{\mathrm{U}} \max_{\ell\in \Phi_\mathrm{M}} L\left(\left|X_{\ell,\mathrm{M}}\right|\right),\\
L^*\left( \left|X_i\right| \right)&=\max_{j\in \Phi_i} L (\left|X_{j,i}\right|).
\end{align}
Here, $G_a^{\mathrm{U}}$ is the power gain of the serving MBS and $L^*\left(\left|X_\mathrm{M}\right|\right)$ can be viewed as compensated path loss due to the power gain.

\subsection{Downlink WPT Model}
{For wireless energy harvesting, the RF signals are interpreted as energy. Therefore, in the massive MIMO macrocell, we adopt the simplest linear MRT beamforming\footnote{{Since there is no interference concern in the downlink power transfer, other beamforming methods involving interference mitigation such as zero-forcing (ZF) will reduce power gain and increase the power consumption of the MBS.}} to direct the RF energy towards its $S$ intended users with equal-time sharing.\footnote{{In this way, user receives the largest transferred power in a short time, which means that the user's battery can be quickly recharged.}}  This suboptimal approach also helps with the analytical tractability. Thus, for each intended user of the macrocell at a communication block time $T$, the directed power transfer time is $\frac{\tau T}{S}$, the isotropic power transfer time is $\frac{(S-1)\tau T}{S}$, and the ambient RF energy from nearby BSs is harvested during the whole energy harvesting time $\tau T$. } We use the short-range propagation  model~\cite{Baccelli2006_gen,Kaibin2014} to avoid singularity caused by proximity between the BSs and the users, {which guarantees that the random distance between user and BS  is larger than a fixed reference distance, and such constraint is also considered in the 3GPP channel model~\cite{3GPP_channel_model}}. This will ensure that users receive finite average power.
{We assume that the RF energy harvesting sensitivity level is very small (e.g. -10 dBm~\cite{Rui_Zhang_2013}) and can be omitted~\cite{Rui_Zhang_2013,Kaibin2014,S_A_H_WPT_2015}. In fact, this paper considers users with large energy storages (which will be specified in the following section) such that  enough harvested energy can be stored for supporting stable transmit power, which implies that the small level of the minimum incident
energy has negligible contribution on the amount of  harvested energy.} {As the energy harvested from the noise is negligible, during the energy harvesting phase, the total harvested energy   at a typical user $o$ that is associated with the MBS is given by
\begin{figure}
    \begin{center}
\includegraphics[height=2.8in,width=4.5in]{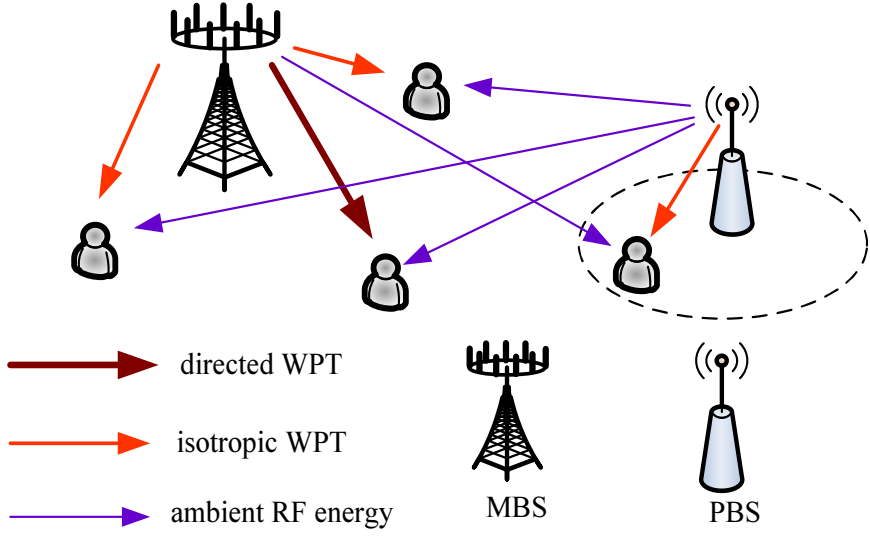}
        \caption{{An illustration of wireless power transfer in the two-tier HetNet consisting of massive MIMO MBS and picocell base station (PBS).}}
        \label{system_model}
    \end{center}
\end{figure}
\begin{multline}\label{MBS_RE_Power}
\mathrm{E}_{o,\mathrm{M}}=\underbrace{\eta{{{P_\mathrm{M}}}}{h_o}L\left( \max\left\{{\left| {{X_{{o},\mathrm{M}}}} \right|},d\right\}
 \right)\times\frac{\tau T}{S}}_{\mathrm{E}_{o,\mathrm{M}}^1}\\
+\underbrace{\eta{{{P_\mathrm{M}}}}{h'_o}L\left( \max\left\{{\left| {{X_{{o},\mathrm{M}}}} \right|},d\right\}
 \right)\times\frac{\left(S-1\right)\tau T}{S}}_{\mathrm{E}_{o,\mathrm{M}}^2}\\
+\underbrace{\eta \left( I_{\mathrm{M},1}\times\tau T+I_{\mathrm{S},1}\times\tau T\right)}_{\mathrm{E}_{o,\mathrm{M}}^3},
\end{multline}}
where $\mathrm{E}_{o,\mathrm{M}}^1$ is the energy from the directed WPT, $\mathrm{E}_{o,\mathrm{M}}^2$ is the energy from the isotropic WPT, and  $\mathrm{E}_{o,\mathrm{M}}^3$ is the energy from the ambient RF, {as illustrated in Fig.~\ref{system_model}}. Here, $0< \eta < 1$ is the RF-to-DC conversion efficiency, $d>0$ denotes the reference distance, $h_o\sim \Gamma\left(N,1\right)$ and $\left| {{X_{{o},\mathrm{M}}}} \right|$ are, respectively, the small-scale fading channel power gain and the distance when the serving MBS recharges the typical user, and ${h'_o}\sim \rm{exp}(1)$ is the small-scale fading channel power gain when the serving MBS directly transfers energy to other users in the same cell. In addition,
{\begin{equation}
I_{\mathrm{M},1}=\sum\limits_{\ell  \in {\Phi_\mathrm{M}}\setminus\left\{ o \right\}} {{P_\mathrm{M}}{h_\ell }L\left(\max \left\{ {\left| {{X_{\ell ,\mathrm{M}}}} \right|},d\right\}\right)}
\end{equation}
is the sum of interference from the interfering MBSs in the first tier, where $h_\ell \sim \Gamma\left(1,1\right)$ }and $\left| {{X_{\ell ,\mathrm{M}}}} \right|$ denote, respectively, the small-scale fading interfering channel gain and the distance between a typical user and MBS $\ell \in {\Phi_\mathrm{M} \setminus\left\{ o \right\}}$ ({except the typical user's serving MBS}), and
\begin{equation}
I_{\mathrm{S},1}=\sum\limits_{i = 2}^K {\sum\limits_{j \in {\Phi _i}} {{P_i}{h_j}L\left( \max \left\{ {\left| {{X_{j,i}}} \right|},d\right\}\right)} }
\end{equation}
is the sum of interference from the SBSs in the first tier, where $h_j \sim \rm{exp}(1)$ and $\left| {{X_{j,i}}} \right|$ are, respectively, the small-scale fading interfering channel power gain and the distance between a typical user and SBS $j  \in {\Phi_i}$. In each power transfer phase, the harvested energy at a typical user $o$ associated with the SBS in the $k$-th tier can also be written as
\begin{align}\label{SBS_RE_Power}
\mathrm{E}_{o,k}=\underbrace{\eta{{{P_k}}}{g_o}L\left( \max\left\{{\left| {{X_{{o},k}}} \right|},d\right\}
 \right)\times\tau T}_{\mathrm{E}_{o,k}^1}
 +\underbrace{\eta \left(I_{\mathrm{M},k}+I_{\mathrm{S},k}\right)\times\tau T}_{\mathrm{E}_{o,k}^2},
\end{align}
where $\mathrm{E}_{o,k}^1$ is the energy from the isotropic WPT and $\mathrm{E}_{o,k}^2$ is the energy from the ambient RF, $g_o\sim \Gamma\left(1,1\right)$ and $\left| {{X_{{o},k}}} \right|$ are the small-scale fading channel power gain and the distance between a typical user and its associated MBS, respectively, and similar to the above, we also have
{\begin{equation}
I_{\mathrm{M},k}=\sum\limits_{\ell  \in {\Phi_\mathrm{M}}} {{P_\mathrm{M}}{g_\ell }L\left(\max \left\{ {\left| {{X_{\ell,\mathrm{M}}}} \right|},d\right\}\right)},
\end{equation}
in which $g_\ell \sim \Gamma\left(1,1\right)$ } and ${\left| {{X_{\ell,\mathrm{M}}}} \right|}$ are, respectively, the small-scale fading interfering channel power gain and the distance between a typical user and MBS $\ell$,  and
\begin{equation}
I_{\mathrm{S},k}=\sum\limits_{i = 2}^K {\sum\limits_{j \in {\Phi _i}\setminus \left\{o\right\}} {{P_i}{g_{j,i}}L\left( \max \left\{ {\left| {{X_{j,i}}} \right|},d\right\}\right)}},
\end{equation}
in which $g_{j,i}\sim \Gamma\left(1,1\right)$ and $\left| {{X_{j,i}}} \right|$ are, respectively, the small-scale fading interfering channel power gain and the distance between a typical user and SBS $j  \in {\Phi_i}\setminus\left\{o\right\}$.

\subsection{Uplink WIT Model}
 After energy harvesting, user $u_i$ transmits information signals to the serving BS with a specific transmit power $P_{u_i}$. In the uplink, each MBS uses linear zero-forcing beamforming (ZFBF) to simultaneously receive $S$ data streams from its  $S$ intended users to cancel the intra-cell interference, which has been widely used in the massive MIMO literature \cite{Hoon2012,Hosseini2014_Massive}.


For a typical user that is associated with its typical serving MBS, the received signal-to-interference-plus-noise ratio (SINR) at its typical serving MBS is given by
\begin{equation}\label{SINR_Macro}
{\rm SINR}_\mathrm{M} = \frac{{{P_{{u_o}}}{h_{o,\mathrm{M}}}L\left( \max \left\{\left| {{X_{o,\mathrm{M}}}} \right|,d\right\} \right)}}{{{I_{u,\mathrm{M}}} + {I_{u,\mathrm{S}}} + {\delta ^2}}},
\end{equation}
where
\begin{equation}
\left\{\begin{aligned}
{I_{u,\mathrm{M}}} &= \sum\limits_{i \in {{\widetilde {\mathcal{U}}}_\mathrm{M}} \setminus \left\{ o \right\}} {{P_{{u_i}}}{h_i}L\left( {\max \left\{ {\left| {{X_i}} \right|,d} \right\}} \right)},\\
{I_{u,\mathrm{S}}}&=\sum\limits_{i = 2}^K {\sum\limits_{j \in {{\widetilde {\mathcal{U}}}_i}} {{P_{{u_j}}}{h_j}L\left( {\max \left\{ {\left| {{X_j}} \right|,d} \right\}} \right)} },
\end{aligned}\right.
\end{equation}
 ${h_{o,\mathrm{M}}}\sim \Gamma\left(N-S+1,1\right)$ \cite{Hosseini2014_Massive} and $\left|X_{o,\mathrm{M}}\right|$ are the small-scale fading channel power gain and the distance between a typical user and its typical serving MBS, respectively, $h_i\sim \rm{exp}(1)$ and $\left|X_i\right|$ are the small-scale fading interfering channel power gain and the distance between the interfering user $u_i$ and the typical serving MBS, respectively, $\tilde{\mathcal{U}}_\mathrm{M}$ is the point process corresponding to the interfering users in the macrocells, while $\tilde{\mathcal{U}}_i$ is the point process corresponding to the interfering users in the $i$-th tier, and $\delta ^2$ denotes the noise power.

Likewise, for a typical user associated with the typical serving SBS in the $k$-th tier, the received SINR is given by
\begin{equation}\label{SINR_Small}
{\rm SINR}_k= \frac{{{P_{{u_o}}}{g_{o,k}}L\left( \max\left\{\left| {{X_{o,k}}} \right|,d\right\} \right)}}{{{I_{u,\mathrm{M}}} + {I_{u,\mathrm{S}}} + {\delta ^2}}},
\end{equation}
where
\begin{equation}
\left\{\begin{aligned}
{I_{u,\mathrm{M}}} &= \sum\limits_{i \in {{\widetilde {\mathcal{U}}}_\mathrm{M}} } {{P_{{u_i}}}{g_i}L\left( {\max \left\{ {\left| {{X_i}} \right|,d} \right\}} \right)},\\
{I_{u,\mathrm{S}}}&=\sum\limits_{i = 2}^K {\sum\limits_{j \in {{\widetilde {\mathcal{U}}}_i}{ \setminus \left\{ o \right\}}} {{P_{{u_j}}}{g_j}L\left( {\max \left\{ {\left| {{X_j}} \right|,d} \right\}} \right)} },
\end{aligned}\right.
\end{equation}
$g_{o,k}\sim \mathrm{exp}(1)$ and $\left| {{X_{o}}} \right|$ are the small-scale fading channel gain and the distance between a typical user and its typical serving SBS, respectively, $g_{i}\sim \mathrm{exp}(1)$ and $\left| {{X_{i}}} \right|$ are the small-scale fading interfering channel gain and the distance between the interfering user $u_i$ and the typical serving BS, respectively.

\section{Energy Analysis}\label{section_WPT}
Here, the average harvested energy is derived assuming that users are equipped with large energy storage so that users can transmit reliably after energy harvesting. Considering the fact that the energy consumed for uplink information transmission should not exceed the harvested energy, the stable transmit power $P_{u_o}$ for a typical user should satisfy~\cite{Kaibin2014}
\begin{align}\label{power_constr_1}
P_{u_o} \leq  \frac{\overline{\mathrm{E}}_{o}}{\left(1-\tau\right) T},
\end{align}
where  $\overline{\mathrm{E}}_{o}$ denotes the average harvested energy.

\subsection{{New Statistical Properties}}
Before deriving the average harvested energy, we find the following lemmas useful.

\begin{lemma}
Under DRSP-based user association, the probability density functions (PDFs) of the distance $ \left|X_{o,\mathrm{M}}\right| $ between a typical user and its serving MBS and the distance $ \left|X_{o,k}\right| $ between a typical user and its serving SBS in the $ k $-th tier are, respectively, given by
\begin{align}\label{DRSP_PDF_M}
f_{\left| {{X_{o,{\mathrm{M}}}}} \right|}^\mathrm{DRSP}(x) = \frac{{2\pi {\lambda _{\mathrm{M}}}}x }{{{\Psi _{\mathrm{M}}^ \mathrm{DRSP}}}}\exp\left(  - \pi {\lambda _{\mathrm{M}}}{x^2} - \pi \sum\limits_{i = 2}^K {{\lambda _i}{{\hat r}_{{\mathrm{M}}\mathrm{S}}}^2 {x^{\frac{2{\alpha _{\mathrm{M}}}}{\alpha _i}}}} \right),
\end{align}
and
\begin{align}\label{Lemma-f-down-S}
f_{\left| {{X_{o,k}}} \right|}^\mathrm{DRSP}(y) =\frac{{2\pi \lambda _{k}}y}{{\Psi _k^\mathrm{DRSP}}}\times
\exp\left(  - \pi {\lambda _{\mathrm{M}}}{{\hat r}_{\mathrm{SM}}^2}{{y^{\frac{2\alpha _k}{\alpha _{\mathrm{M}}}}}} - \pi \sum\limits_{i = 2}^K {{\lambda _i}{{\hat r}_{\mathrm{SS}}^2}{y^{\frac{2\alpha _k}{\alpha _i}}}} \right),
\end{align}
in which ${\hat r}_{{\mathrm{M}}\mathrm{S}} = {\left( {G_a^{\mathrm{D}}\frac{{{P_{\mathrm{M}}}}}{{S{P_i}}}} \right)^{ \frac{- 1}{\alpha _i}}}$ with $G_a^{\mathrm{D}}=(N + S - 1)$, ${\hat r_{\mathrm{S}{\mathrm{M}}}} = {\left( {\frac{{{SP_k}}}{{G_a^{\mathrm{D}}{P_{\mathrm{M}}}}}} \right)^{ \frac{- 1}{\alpha _\mathrm{M}}}}$, and ${{\hat r}_{\mathrm{SS}}}= {\left( {\frac{{{P_k}}}{{{P_i}}}} \right)^{ \frac{- 1}{\alpha _i}}}$. Also, in (\ref{DRSP_PDF_M}), $ \Psi_\mathrm{M}^ \mathrm{DRSP} $ is the probability that a typical user is associated with the MBS, given by
\begin{align}\label{M_Pr}
\Psi_\mathrm{M}^ \mathrm{DRSP}=2 \pi  \lambda_{\mathrm{M}}\times
 \int_0^\infty  r \exp\left(- \pi \lambda_\mathrm{M} r^2-\pi \sum\limits_{i = 2}^K {\lambda_i {{\hat r}_{\mathrm{M}\mathrm{S}}^2}{r^{\frac{2{\alpha _{\mathrm{M}}}}{\alpha _i}}}}  \right)dr,
\end{align}
and $ \Psi_k^ \mathrm{DRSP} $ is the probability that a typical user is associated with the SBS in the $k$-th tier, which is given by
\begin{align}\label{DRSP_k_prob}
\hspace{-0.3 cm} \Psi_k^\mathrm{DRSP} = 2\pi {\lambda _k}\times
\int_0^\infty r ~{\exp} \left(  - \pi {\lambda _{\mathrm{M}}}{{\hat r}_{\mathrm{SM}}^2}{{r^{\frac{2\alpha _k}{\alpha _{\mathrm{M}}}}}} - \pi \sum\limits_{i = 2}^K {{\lambda _i}{{\hat r}_{\mathrm{SS}}^2}{r^{\frac{2\alpha _k}{\alpha _i}}}} \right)dr.
\end{align}
\end{lemma}

\begin{proof}
See Appendix~A.
\end{proof}

Based on \eqref{M_Pr}, we obtain a simplified asymptotic expression for the probability in the following corollary.

\begin{corollary}
For large number of antennas with $N \to \infty $, using the Taylor series expansion truncated to the first order, the probability that a typical user is associated with the MBS given by \eqref{M_Pr} is asymptotically derived as
\begin{align}
\Psi_{\mathrm{M}_\infty}^ \mathrm{DRSP} = 2\pi {\lambda _{\mathrm{M}}}\times
\left(\begin{array}{c}
\int_0^\infty  r\exp \left(- \pi \lambda _{\mathrm{M}}{r^2}\right)dr
- \pi \sum\limits_{i = 2}^K {\lambda _i}{\hat r}_{\mathrm{M}\mathrm{S}}^2\int_0^\infty  r^{1 + \frac{2\alpha _{\mathrm{M}}}{\alpha _i}}\exp \left( { - \pi {\lambda _{\mathrm{M}}}{r^2}} \right)dr
\end{array}\right),
\end{align}
which can be expressed as
\begin{equation}\label{M_Pr_asym}
\Psi_{\mathrm{M}_\infty}^ \mathrm{DRSP}= 1 - \pi \sum\limits_{i = 2}^K {{\lambda _i}{{{{\hat r}_{\mathrm{M}\mathrm{S}}^2}}}\frac{{\Gamma \left( {1 + \frac{\alpha _{\mathrm{M}}}{\alpha _i}} \right)}}{{{{\left( {\pi {\lambda _{\mathrm{M}}}} \right)}^{\frac{\alpha _{\mathrm{M}}}{\alpha _i}}}}}}.
\end{equation}
Note that the probability for a user associated with the SBS is $ 1- \Psi_{\mathrm{M}_\infty}^ \mathrm{DRSP}$. From \eqref{M_Pr_asym}, it is explicitly shown that the probability for a user associated with the MBS increases with the density of MBS but decreases with the density of SBS.
\end{corollary}

Likewise, in the case of the URSP-based user association, we have the following lemma and corollary. As the approaches are similar, their proofs are omitted.

\begin{lemma}
Under URSP-based user association, the PDFs of the distance $ \left|X_{o,\mathrm{M}}\right| $ between a typical user and its serving MBS and the distance $ \left|X_{o,k}\right| $ between a typical user and its serving SBS in the $ k $-th tier are, respectively, given by
\begin{align}\label{Lemma-f-up-M}
f_{\left| {{X_{o,{\mathrm{M}}}}} \right|}^\mathrm{URSP}(x) = \frac{{2\pi x}}{{\Psi _{\mathrm{M}}^\mathrm{URSP}}}{\lambda _{\mathrm{M}}}\times
\exp\left(  - \pi {\lambda _{\mathrm{M}}}{x^2} - \pi \sum\limits_{i = 2}^K {{\lambda _i}{{\widetilde r}_\mathrm{MS}^2}{x^{\frac{2 \alpha_\mathrm{M}}{\alpha_i}} }} \right),
\end{align}
and
\begin{align}\label{Lemma-f-up-S}
f_{\left| {{X_{o,k}}} \right|}^\mathrm{URSP}(y) = \frac{{2\pi y}}{{\Psi _k^\mathrm{URSP}}}{\lambda _{k}}\times
\exp\left(  - \pi {\lambda _{\mathrm{M}}}{\widetilde r_\mathrm{SM}^2}{y^{\frac{2 \alpha_k}{\alpha_{\mathrm{M}}}}} - \pi \sum\limits_{i = 2}^K {{\lambda _i}{y^{\frac{2 \alpha_k}{\alpha_i}}}} \right),
\end{align}
where ${{\widetilde r}_\mathrm{MS}} = {\left (G_a^{\mathrm{U}} \right)^{\frac{-1}{\alpha_i}}}$ with $G_a^{\mathrm{U}}=(N-S+1)$, and ${\widetilde r}_{\mathrm{SM}} = \left( {\frac{1}{{G_a^{\mathrm{U}}}}} \right)^{\frac{-1}{\alpha_\mathrm{M}}}$. Also, in the above expressions, we have
\begin{align}\label{U_M_Pr}
\Psi_\mathrm{M}^ \mathrm{URSP}=2 \pi  \lambda_{\mathrm{M}} \times
\int_0^\infty  r \exp\left(- \pi \lambda_\mathrm{M} r^2-\pi \sum\limits_{i = 2}^K {\lambda_i {{\widetilde r}_{\mathrm{M}\mathrm{S}}^2}{r^{\frac{2{\alpha _{\mathrm{M}}}}{\alpha _i}}}}  \right)dr,
\end{align}
and
\begin{align}\label{URSP_k_prob}
\Psi_k^\mathrm{URSP} = 2\pi {\lambda _k}\times
\int_0^\infty r {\exp} \left(  - \pi {\lambda _{\mathrm{M}}}{{\widetilde r}_{\mathrm{SM}}^2}{{r^{\frac{2\alpha _k}{\alpha _{\mathrm{M}}}}}} - \pi \sum\limits_{i = 2}^K {{\lambda _i}{r^{\frac{2\alpha _k}{\alpha _i}}}} \right)dr.
\end{align}
\end{lemma}

\begin{corollary}
For URSP-based user association, with large $N$, the asymptotic expression for the probability that a typical user is associated with the MBS given by \eqref{U_M_Pr} can be expressed as
\begin{align}\label{U_M_Pr_asym}
\Psi_{\mathrm{M}_\infty}^ \mathrm{URSP} =  1 - \pi \sum\limits_{i = 2}^K {{\lambda _i}{{{{\widetilde r}_{\mathrm{M}\mathrm{S}}^2}}}\frac{{\Gamma \left( {1 + \frac{\alpha _{\mathrm{M}}}{\alpha _i}} \right)}}{{{{\left( {\pi {\lambda _{\mathrm{M}}}} \right)}^{\frac{\alpha _{\mathrm{M}}}{\alpha _i}}}}}}.
\end{align}
In addition, the probability that a user is associated with the SBS can be directly found by $1-\Psi_{\mathrm{M}_\infty}^ \mathrm{URSP}$.
\end{corollary}

\subsection{{Average Harvested Energy}}
Using DRSP-based user association, the maximum average harvested energy can be achieved. Here, we first derive the conditional expression of the average harvested energy given the distance between a typical user and its serving BS.

\begin{figure*}[!t]
\normalsize
\begin{multline}\label{DRSP_con_AverEH_M}
\widetilde{\mathrm{E}}_{o,\mathrm{M}}^{\mathrm{DRSP}}\left(x\right)= \eta \Bigg\{\left({N+S-1}\right) {{\frac{P_\mathrm{M}}{S}}}\beta \left({\rm{\mathbf{1}}}\left( x \leq d \right) d^{-\alpha_\mathrm{M}}+{\rm{\mathbf{1}}}\left( x > d \right) x^{-\alpha_\mathrm{M}}\right)\\
+P_\mathrm{M} \beta 2 \pi {\lambda _{\mathrm{M}}} \left({\rm{\mathbf{1}}}\left( x \leq d \right) \left({d^{ - {\alpha _{\mathrm{M}}}}} \frac{{({d^2} - {x^2})}}{2} - \frac{{{d^{2 - {\alpha _{\mathrm{M}}}}}}}{{2 - {\alpha _{\mathrm{M}}}}}\right) -{\rm{\mathbf{1}}}\left( x > d \right) \frac{{{x^{2 - {\alpha _{\mathrm{M}}}}}}}{{2 - {\alpha _{\mathrm{M}}}}} \right)\\
+\sum\limits_{i = 2}^K P_i \beta 2 \pi {\lambda _{i}}\left.\left({\rm{\mathbf{1}}}\left( x \leq d_o \right) \left( {d^{- {\alpha _i}}}\frac{{\left({d^2} - {{\hat r}_{{\mathrm{M}}\mathrm{S}}}^2 {x^{\frac{2{\alpha _{\mathrm{M}}}}{\alpha _i}}}\right)}}{2} - \frac{{{d^{2 - {\alpha _i}}}}}{{2 - {\alpha _i}}} \right)-{\rm{\mathbf{1}}}\left( x > d_o \right) \frac{{{\hat r}_{{\mathrm{M}}\mathrm{S}}}^{\left(2-\alpha _i\right)} {x^{\frac{{\alpha _{\mathrm{M}}}(2-\alpha _i) }{\alpha _i}}}}{2-\alpha _i} \right) \right\}\times\tau T,
\end{multline}
\hrulefill 
\end{figure*}
\begin{figure*}[!t]
\normalsize
\begin{multline}\label{DRSP_con_AverEH_k}
\widetilde{\mathrm{E}}_{o,k}^{\mathrm{DRSP}}\left(y\right)= \eta \Bigg\{ {P_k} \beta  \left({\rm{\mathbf{1}}}\left( y\leq d \right) d^{-\alpha_k}+{\rm{\mathbf{1}}}\left( y > d \right) y^{-\alpha_k}\right)\\
+P_\mathrm{M} \beta 2 \pi {\lambda _{\mathrm{M}}} \left({\rm{\mathbf{1}}}\left( y \leq d_1 \right) \left({d^{ - {\alpha _{\rm M}}}}{\frac{{\left({d^2} - {{\hat r}_{\mathrm{SM}}^2}{{y^{\frac{2\alpha _k}{\alpha _{\mathrm{M}}}}}}    \right)}}{2} - \frac{{{d^{2 - {\alpha _{\rm M}}}}}}{{2 - {\alpha _{\rm M}}}}}\right)- {\rm{\mathbf{1}}}\left( y > d_1 \right)  \frac{{\hat r}_{\mathrm{SM}}^{2-\alpha _{\rm M}} y^{\frac{\alpha _k\left(2-\alpha _{\rm M}\right)}{\alpha _{\mathrm{M}}}} }{{2 - {\alpha _{\rm M}}}} \right)\\
+\sum\limits_{i = 2}^K  \beta 2 \pi {\lambda_i}\left. \left({\rm{\mathbf{1}}}\left( y \leq d_2 \right) \left({{d^{ - {\alpha _i}}} \frac{{\left({d^2} - {{\hat r}_{\mathrm{SS}}^2}y^{\frac{2 \alpha _k}{\alpha _i}} \right)}}{2} - \frac{{{d^{2 - {\alpha _i}}}}}{{2 - {\alpha _i}}}}  \right) - {\rm{\mathbf{1}}}\left( y > d_2 \right) \frac{{{\hat r}_{\mathrm{SS}}}^{2 - {\alpha _i}}y^{\frac{ \alpha _k \left(2 - {\alpha _i}\right)}{\alpha _i}}  }{{2 - {\alpha _i}}}  \right) \right\}\times\tau T,
\end{multline}
\hrulefill
\end{figure*}
\begin{figure*}[!t]
\normalsize
\begin{multline}\label{DRSP_AverEH_M_asymptotic}
\overline{\mathrm{E}}_{o,{\mathrm{M}_\infty}}^{\mathrm{DRSP}}=\eta \Bigg\{\left({N+S-1}\right) {{\frac{P_\mathrm{M}}{S}}}\beta \left(\Xi_1\left(d\right) d^{-\alpha_\mathrm{M}}+\Xi_2 \left( {d,-\alpha_\mathrm{M}} \right)\right)\\
+P_\mathrm{M} \beta 2 \pi {\lambda _{\mathrm{M}}} \left(d^{2- {\alpha _{\mathrm{M}}}}\frac{\alpha _{\mathrm{M}}}{2\left(\alpha _{\mathrm{M}}-2\right)}\Xi_1\left(d\right)-\frac{d^{ - {\alpha _{\mathrm{M}}}}}{2}\Xi_3 \left( {d,2} \right)+\frac{\Xi_2 \left(d,2-\alpha _{\mathrm{M}}\right)}{\alpha _{\mathrm{M}}-2}\right)+\sum\limits_{i = 2}^K P_i \beta 2 \pi {\lambda _{i}}\\
\times \left.\left( d^{2-\alpha_i}\frac{\alpha_i}{2\left(\alpha_i-2\right)} \Xi_1\left(d_o\right)-\frac{d^{-\alpha_i} {{\hat r}_{{\mathrm{M}}\mathrm{S}}}^2 }{2} \Xi_3 \left( d_o,\frac{2{\alpha _{\mathrm{M}}}}{\alpha _i}\right) + \frac{{{\hat r}_{{\mathrm{M}}\mathrm{S}}}^{\left(2-\alpha _i\right)}}{\alpha _i-2} \Xi_2 \left( d_o, {\frac{{\alpha _{\mathrm{M}}}(2-\alpha _i) }{\alpha _i}}\right) \right)\right\}\times\tau T,
\end{multline}
\hrulefill 
\end{figure*}

\begin{theorem}
For the case of DRSP-based user association, given the distances $ \left|X_{o,\mathrm{M}}\right|=x$ and $ \left|X_{o,k}\right|=y$, the conditional expressions of the average harvested energy for a typical user that is associated with an MBS and that for a typical user that is associated with an SBS in the $k$-th tier are, respectively, given by \eqref{DRSP_con_AverEH_M} and \eqref{DRSP_con_AverEH_k} at the top of next page,  $d_o=\left({{\hat r}_{{\mathrm{M}}\mathrm{S}}}\right)^{-\frac{\alpha _i}{\alpha_{\mathrm{M}}}}d^{{\alpha _i}/{\alpha _{\mathrm{M}}}}$, $d_1={\left({\hat r}_{\mathrm{SM}}\right)^{\frac{-\alpha _{\rm M}}{\alpha _k}}}{d^{{\alpha _{\rm M}}/{\alpha _k}}}$, and $d_2={\left({{\hat r}_{\mathrm{SS}}}\right)^{\frac{-\alpha _i}{\alpha _k}}}{d^{{\alpha _i}/{\alpha _{k}}}}$.
\end{theorem}

\begin{proof}
See Appendix~B.
\end{proof}

Based on Theorem 1, the average harvested energy for a user that is associated with an MBS and that a user that is associated with an SBS in the $k$-th tier are found as
\begin{align}\label{DRSP_AverEH_M}
&\overline{\mathrm{E}}_{o,\mathrm{M}}^{\mathrm{DRSP}}=\int_0^\infty \widetilde{\mathrm{E}}_{o,\mathrm{M}}^{\mathrm{DRSP}}\left(x\right)  f_{\left| {{X_{o,{\mathrm{M}}}}} \right|}^\mathrm{DRSP}(x) dx,
\end{align}
and
\begin{align}\label{DRSP_AverEH_k}
&\overline{\mathrm{E}}_{o,k}^{\mathrm{DRSP}}=\int_0^\infty \widetilde{\mathrm{E}}_{o,k}^{\mathrm{DRSP}}\left(y\right)  f_{\left| {{X_{o,k}}} \right|}^\mathrm{DRSP}(y) dy.
\end{align}

\begin{corollary}
When the number of antennas at the MBS grows large, we obtain the asymptotic expression for $\overline{\mathrm{E}}_{o,\mathrm{M}}^{\mathrm{DRSP}}$ in \eqref{DRSP_AverEH_M}  as \eqref{DRSP_AverEH_M_asymptotic} (see next page), where $\Xi_1(\cdot)$, $\Xi_2 \left( \cdot,\cdot\right)$ and $\Xi_3 \left( \cdot,\cdot\right)$ are, respectively, given by
\begin{align}\label{asym_CDF_exp_M}
\Xi_1(x) = \frac{1}{\Psi_{\mathrm{M}_\infty}^ \mathrm{DRSP}}\times
\left(1-e^{ - \pi {\lambda_{\rm{M}}}{x^2}}-\pi \sum\limits_{i = 2}^K {{\lambda _i}} \hat r_{{\rm{MS}}}^2
\frac{{\gamma \left( {1 + \frac{\alpha _{\rm{M}}}{\alpha _i},\pi {\lambda _{\rm{M}}}{x^2}} \right)}}{{{{\left( {\pi {\lambda _{\rm{M}}}} \right)}^{\frac{\alpha _{\rm{M}}}{\alpha _i}}}}}     \right),
\end{align}
\begin{align}\label{E_a_b_function}
\Xi_2 \left( {a{\rm{,}}b} \right)=\frac{1}{{\Psi _{{{\rm{M}}_\infty }}^{{\rm{DRSP}}}}}\left( \frac{{\Gamma \left( {1 + \frac{b}{2},\pi {\lambda _{\rm{M}}}{a^2}} \right)}}{{{{\left( {\pi {\lambda _{\rm{M}}}} \right)}^{\frac{b}{2}}}}}-\right.
\left.\pi \sum\limits_{i = 2}^K {{\lambda _i}} \hat r_{{\rm{MS}}}^2
\frac{{\Gamma \left( {1 + \frac{{{\alpha _{\rm{M}}}}}{{{\alpha _i}}} + \frac{b}{2},\pi {\lambda _{\rm{M}}}{a^2}} \right)}}{{{{\left( {\pi {\lambda _{\rm{M}}}} \right)}^{\frac{\alpha _{\rm{M}}}{\alpha _i} + \frac{b}{2}}}}} \right),
\end{align}
and
\begin{align}\label{E_a_b2_function}
\Xi_3 \left( {c{\rm{,}}d} \right)=\frac{1}{{\Psi _{{{\rm{M}}_\infty }}^{{\rm{DRSP}}}}}\left( \frac{{\gamma \left( {1 + \frac{d}{2},\pi {\lambda _{\rm{M}}}{c^2}} \right)}}{{{{\left( {\pi {\lambda _{\rm{M}}}} \right)}^{\frac{d}{2}}}}} -\right.
\left.\pi \sum\limits_{i = 2}^K {{\lambda _i}} \hat r_{{\rm{MS}}}^2
\frac{{\gamma \left( {1 + \frac{{{\alpha _{\rm{M}}}}}{{{\alpha _i}}} + \frac{d}{2},\pi {\lambda _{\rm{M}}}{c^2}} \right)}}{{{{\left( {\pi {\lambda _{\rm{M}}}} \right)}^{\frac{\alpha _{\rm{M}}}{\alpha _i} + \frac{d}{2}}}}} \right),
\end{align}
where $\gamma\left(\cdot,\cdot\right)$ and  $\Gamma\left(\cdot,\cdot\right)$ are the upper and lower incomplete gamma functions, respectively~\cite[(8.350)]{gradshteyn}.
\end{corollary}

\begin{proof}
See Appendix~C.
\end{proof}

Overall, for a user in the massive MIMO aided HetNets with DRSP-based user association, its average harvested energy can be calculated as
{\begin{align}\label{downlink_hcn_energy_drsp}
\overline{\mathrm{E}}_{o,\mathrm{HetNet}}^{\mathrm{DRSP}} ={\Psi_\mathrm{M}^{\mathrm{DRSP}}} \overline{\mathrm{E}}_{o,\mathrm{M}}^{\mathrm{DRSP}} + \sum\limits_{k = 2}^K \Psi_k^\mathrm{DRSP}\overline{\mathrm{E}}_{o,k}^{\mathrm{DRSP}}.
\end{align}}

Similarly, for the case of URSP-based user association, the average harvested energy for a typical user that is associated with an MBS and that for a typical user that is associated with an SBS in the $k$-th tier are, respectively, given by
\begin{align}\label{URSP_AverEH_M}
\overline{\mathrm{E}}_{o,\mathrm{M}}^{\mathrm{URSP}}=\int_0^\infty \widetilde{\mathrm{E}}_{o,\mathrm{M}}^{\mathrm{URSP}}\left(x\right)  f_{\left| {{X_{o,{\mathrm{M}}}}} \right|}^\mathrm{URSP}(x) dx,
\end{align}
and
\begin{align}\label{URSP_AverEH_k}
\overline{\mathrm{E}}_{o,k}^{\mathrm{URSP}}=\int_0^\infty \widetilde{\mathrm{E}}_{o,k}^{\mathrm{URSP}}\left(y\right)  f_{\left| {{X_{o,k}}} \right|}^\mathrm{URSP}(y) dy,
\end{align}
where $\widetilde{\mathrm{E}}_{o,\mathrm{M}}^{\mathrm{URSP}}\left(x\right)$ and  $\widetilde{\mathrm{E}}_{o,k}^{\mathrm{URSP}}\left(y\right)$ are obtained by interchanging the parameters  ${\hat r}_{{\mathrm{M}}\mathrm{S}}\rightarrow {{\widetilde r}_\mathrm{MS}}$, ${\hat r}_{{\mathrm{S}}\mathrm{M}}\rightarrow {{\widetilde r}_\mathrm{SM}}$ and ${\hat r}_{{\mathrm{S}}\mathrm{S}} \rightarrow 1$ in \eqref{DRSP_con_AverEH_M} and \eqref{DRSP_con_AverEH_k}, respectively, $f_{\left| {{X_{o,{\mathrm{M}}}}} \right|}^\mathrm{URSP}(x)$ and $f_{\left| {{X_{o,k}}} \right|}^\mathrm{URSP}(y)$ are given by \eqref{Lemma-f-up-M} and \eqref{Lemma-f-up-S}, respectively.

\begin{corollary}
If the number of antennas at the MBS is large for URSP-based user association, then we obtain the asymptotic expression for $\overline{\mathrm{E}}_{o,{\mathrm{M}}}^{\mathrm{URSP}}$ by interchanging $ {\Psi _{{{\rm{M}}_\infty }}^{{\rm{DRSP}}}} \rightarrow \Psi_{\mathrm{M}_\infty}^ \mathrm{URSP} $ and ${\hat r}_{{\mathrm{M}}\mathrm{S}}\rightarrow {{\widetilde r}_\mathrm{MS}}$ in \eqref{DRSP_AverEH_M_asymptotic}.
\end{corollary}

Overall, for a user in the massive MIMO aided HetNets with URSP-based user association, its average harvested energy is calculated as
\begin{align}\label{downlink_hcn_energy_ursp}
\overline{\mathrm{E}}_{o,\mathrm{HetNet}}^{\mathrm{URSP}} ={\Psi_\mathrm{M}^{\mathrm{URSP}}} \overline{\mathrm{E}}_{o,\mathrm{M}}^{\mathrm{URSP}} + \sum\limits_{k = 2}^K \Psi_k^\mathrm{URSP}\overline{\mathrm{E}}_{o,k}^{\mathrm{URSP}}.
\end{align}

\section{Uplink Performance Evaluation}\label{Section_uplink_rate}
{After harvesting the energy, users transmit their messages to the serving BSs with a stable transmit power constrained by \eqref{power_constr_1}.\footnote{{It is indicated from \eqref{power_constr_1} that the power transfer time allocation factor $\tau$ has to be large enough, in order to avoid the power outage.}}} In this section, we analyze the uplink WIT performance in terms of average achievable rate. On the one hand, given a specific user's transmit power, URSP-based user association outperforms the DRSP-based in the uplink by maximizing the uplink received signal power. On the other hand,  compared to URSP-based user association, DRSP-based user association allows users to set a higher stable transmit power due to more harvested energy. Thus, it is necessary to evaluate the uplink achievable rate under these two user association schemes.

We assume that each user intends to set the maximum stable transmit power to achieve the maximum achievable rate. For DRSP-based user association, the transmit power for user $i$ in a macrocell is $P_{u_i}^{\mathrm{DRSP}}=P_{u_\mathrm{M}}^{\mathrm{DRSP}}=\frac{\overline{\mathrm{E}}_{o,\mathrm{M}}^{\mathrm{DRSP}}}{\left(1-\tau\right) T}$, and the transmit power for user $j$ in a small cell of the $k$-th tier is $P_{u_j}^{\mathrm{DRSP}}=P_{u_k}^{\mathrm{DRSP}}=\frac{\overline{\mathrm{E}}_{o,k}^{\mathrm{DRSP}}}{\left(1-\tau\right) T}$, where $\overline{\mathrm{E}}_{o,\mathrm{M}}^{\mathrm{DRSP}}$ and $\overline{\mathrm{E}}_{o,k}^{\mathrm{DRSP}}$ are given by \eqref{DRSP_AverEH_M} and \eqref{DRSP_AverEH_k}, respectively. For URSP-based user association, the transmit power for user $i$ in a macrocell is $P_{u_i}^{\mathrm{URSP}}=P_{u_\mathrm{M}}^{\mathrm{URSP}}=\frac{\overline{\mathrm{E}}_{o,\mathrm{M}}^{\mathrm{URSP}}}{\left(1-\tau\right) T}$, and the transmit power for user $j$ in a small cell of the $k$-th tier is $P_{u_j}^{\mathrm{URSP}}=P_{u_k}^{\mathrm{URSP}}=\frac{\overline{\mathrm{E}}_{o,k}^{\mathrm{URSP}}}{\left(1-\tau\right) T}$, in which $\overline{\mathrm{E}}_{o,\mathrm{M}}^{\mathrm{URSP}}$ and $\overline{\mathrm{E}}_{o,k}^{\mathrm{URSP}}$ are given by \eqref{URSP_AverEH_M} and \eqref{URSP_AverEH_k}, respectively.

\subsection{{Average Uplink Achievable Rate}}
 We first present the achievable rate for the massive MIMO HetNet uplink with DRSP-based user association and have the following theorems.

\begin{theorem}
Given a distance $\left|X_{o,\mathrm{M}}\right|=x$, a tractable lower bound for the conditional average uplink achievable rate between a typical user and its serving MBS can be found as
\begin{align}\label{theo_2}
R_{\mathrm{DRSP},\mathrm{M}}^\mathrm{low}\left(x\right)=
\left(1-\tau\right)\log_2
 \left(1+{P_{{u_\mathrm{M}}}^{\mathrm{DRSP}}}{\left(N-S+1\right)}\frac{\Delta_1\left(x\right)}{\Lambda_\mathrm{DRSP}}\right),
\end{align}
where $\Delta_1\left(x\right)=\beta\left({\rm{\mathbf{1}}}\left( x \leq d \right)d^{-\alpha _{\rm M}}+{\rm{\mathbf{1}}}\left( x > d \right)x^{-\alpha _{\rm M}}\right)$ and
\begin{align}
\Lambda _{\rm DRSP} =2\pi\beta  \left( {P_{{u_{\rm M}}}^{\rm DRSP} (S{\lambda _{\rm M}}) + \sum\limits_{i = 2}^K {P_{{u_i}}^{\rm DRSP}}  {\lambda _i}} \right)
\times \left( {\frac{{{d^{2 - {\alpha _{\rm M}}}}}}{2} + \frac{{{d^{2 - {\alpha _{\rm M}}}}}}{{{\alpha _{\rm M}} - 2}}} \right) + {\delta ^2}.
\end{align}
\end{theorem}

\begin{proof}
See Appendix D.
\end{proof}

\begin{theorem}
Given a distance $\left|X_{o,k}\right|=y$, the conditional average uplink achievable rate between a typical user and its serving SBS in the $k$-th tier is given by
\begin{equation}\label{aver_rate_uplink}
R_{\mathrm{DRSP},k}\left(y\right)=\frac{\left(1-\tau\right)}{\ln 2}\int_0^\infty \frac{{\bar{F}}_{\mathrm{SINR}}\left(x\right)}{1+x} dx,
\end{equation}
where
\begin{equation}
{\bar{F}}_{\mathrm{SINR}}\left(x\right)=e^{-\frac{x {\delta ^2}}{P_{u_k}^{\mathrm{DRSP}}\Delta_2\left(y\right)}-\Omega\left(\frac{x }{P_{u_k}^{\mathrm{DRSP}}\Delta_2\left(y\right)}\right)}
\end{equation}
is the complementary cumulative distribution function (CCDF) of the received SINR, in which
\begin{equation}
\Delta_2\left(y\right)=\beta\left({\rm{\mathbf{1}}}\left( y \leq d \right)d^{-\alpha _{k}}+{\rm{\mathbf{1}}}\left( y > d \right)x^{-\alpha _{k}}\right),
\end{equation}
and $\Omega\left(\cdot\right)$ is given by \eqref{CCDF_tier_k} (see next page). In \eqref{CCDF_tier_k}, $_2{F_1}\left[\cdot,\cdot;\cdot;\cdot\right]$ is the Gauss hypergeometric function~\cite[(9.142)]{gradshteyn}.
\end{theorem}

\begin{figure*}[!t]
\normalsize
\begin{multline}\label{CCDF_tier_k}
\Omega\left(s\right)= \pi (S\lambda_\mathrm{M}) \frac{s P_{u_\mathrm{M}}^{\mathrm{DRSP}}\beta d^{-\alpha_i}}{1+s P_{u_\mathrm{M}}^{\mathrm{DRSP}} \beta d^{-\alpha_i}} d^2 +2\pi(S\lambda_\mathrm{M}) s P_{u_\mathrm{M}}^{\mathrm{DRSP}} \beta  \frac{ d^{2-\alpha_i}}{\alpha_i-2} {}_2{F_1}\left[1,\frac{{{\alpha_i} - 2}}{{{\alpha _i}}};2 - \frac{2}{{{\alpha _i}}}; - s P_{u_\mathrm{M}}^{\mathrm{DRSP}}\beta {{d}^{ - {\alpha _i}}} \right] \\
+\sum\limits_{i = 2}^K \pi \lambda_i  \frac{s P_{u_i}^{\mathrm{DRSP}}\beta d^{-\alpha_i}}{1+s P_{u_i}^{\mathrm{DRSP}} \beta d^{-\alpha_i}} d^2 +\sum\limits_{i = 2}^K 2\pi\lambda_i s P_{u_i}^{\mathrm{DRSP}} \beta \frac{ d^{2-\alpha_i}}{\alpha_i-2} {}_2{F_1}
\left[1,\frac{{{\alpha _i} - 2}}{{{\alpha _i}}};2 - \frac{2}{{{\alpha _i}}}; - s P_{u_i}^{\mathrm{DRSP}}\beta {{d}^{ - {\alpha _i}}} \right]
\end{multline}
\hrulefill 
\end{figure*}

\begin{proof}
See Appendix E.
\end{proof}

With the help of Theorem 2 and Theorem 3, the lower bound for the average uplink achievable rate between a typical user and its serving MBS can be expressed as
\begin{align}\label{aver_up_rate_DRSP}
{\overline R}_{\mathrm{DRSP},\mathrm{M}}^\mathrm{low}=\int_0^\infty{R_{\mathrm{DRSP},\mathrm{M}}^\mathrm{low}\left(x\right)
 f_{\left| {{X_{o,{\mathrm{M}}}}} \right|}^\mathrm{DRSP}(x) dx},
\end{align}
and the average uplink achievable rate between a typical user and its serving SBS in the $k$-th tier is given by
\begin{align}\label{aver_rate_uplink}
{\overline{R}}_{\mathrm{DRSP},k}=\int_0^\infty  R_{\mathrm{DRSP},k} \left(y\right)  f_{\left| {{X_{o,k}}} \right|}^\mathrm{DRSP}(y) dy.
\end{align}

Overall, a lower bound on the average uplink achievable rate for a user in the massive MIMO aided HetNets with DRSP-based user association is calculated as
\begin{align}\label{HCN_uplink_rate_drsp}
{\overline R}_{\mathrm{DRSP},\mathrm{HetNet}}^\mathrm{low}={\Psi_\mathrm{M}^{\mathrm{DRSP}}} {\overline R}_{\mathrm{DRSP},\mathrm{M}}^\mathrm{low} + \sum\limits_{k = 2}^K \Psi_k^\mathrm{DRSP}{\overline{R}}_{\mathrm{DRSP},k}.
\end{align}

For URSP-based user association, the lower bound for the average uplink achievable rate between a typical user and its serving MBS ${\overline R}_{\mathrm{URSP},\mathrm{M}}^\mathrm{low} $ can be directly determined by interchanging the transmit power parameters $P_{u_\mathrm{M}}^{\mathrm{DRSP}}\rightarrow P_{u_\mathrm{M}}^{\mathrm{URSP}}$, $P_{u_i}^{\mathrm{DRSP}}\rightarrow P_{u_i}^{\mathrm{URSP}}$, and the PDF $f_{\left| {{X_{o,{\mathrm{M}}}}} \right|}^\mathrm{DRSP}(x) \rightarrow f_{\left| {{X_{o,{\mathrm{M}}}}} \right|}^\mathrm{URSP}(x)$  in \eqref{aver_up_rate_DRSP}, and the average uplink achievable rate between a typical user and its serving SBS in the $k$-th tier ${\overline{R}}_{\mathrm{URSP},k}$ is obtained by interchanging the transmit power parameters $P_{u_\mathrm{M}}^{\mathrm{DRSP}}\rightarrow P_{u_\mathrm{M}}^{\mathrm{URSP}}$, $P_{u_i}^{\mathrm{DRSP}}\rightarrow P_{u_i}^{\mathrm{URSP}}$, and the PDF $f_{\left| {{X_{o,k}}} \right|}^\mathrm{DRSP}(y) \rightarrow f_{\left| {{X_{o,k}}} \right|}^\mathrm{URSP}(y)$ in \eqref{aver_rate_uplink}. As such, a lower bound on the average uplink achievable rate for a user in the massive MIMO aided HetNets with URSP-based user association is obtained as
\begin{align}\label{HCN_uplink_rate_Ursp}
{\overline R}_{\mathrm{URSP},\mathrm{HetNet}}^\mathrm{low}={\Psi_\mathrm{M}^{\mathrm{URSP}}} {\overline R}_{\mathrm{URSP},\mathrm{M}}^\mathrm{low} + \sum\limits_{k = 2}^K \Psi_k^\mathrm{URSP}{\overline{R}}_{\mathrm{URSP},k}.
\end{align}

%

\section{Numerical Results}
In this section, we present numerical results to examine the impact of different user association schemes and key system parameters on the harvested energy and the uplink achievable rate. We consider a two-tier HetNet consisting of macrocells and picocells. The network is assumed to operate at $f_c=1$ GHz ( $f_c$ is the carrier frequency); the bandwidth (BW) is assumed $10{\rm MHz}$, the density of MBSs is $\lambda_\mathrm{M} = 10^{-3} $ m$^{-2}$ {\footnote{{So far, the number of massive MIMO enabled BSs deployed in the future 5G networks has not been standardized yet.}}}; the density of pico BSs (PBSs) $\lambda_2$ is proportional to $\lambda_\mathrm{M}$; the MBS's transmit power is $ P_\mathrm{M}= 46$ dBm; the noise figure is $\mathrm{Nf}= 10$ dB, the noise power is $ \sigma^2  =-170+10\log_{10}(\mathrm{BW})+ \mathrm{Nf} =-90 $ dBm; {the frequency dependent value $\beta={(\frac{{\text{c}}}{{4\pi {f_c}}})^2}$ with $c=3 \times 10^8 \rm m/s$}; the reference distance $d=1$; and the energy conversion efficiency is $\eta=0.9$. {{{Note that varying the energy conversion efficiency only scales the resulting figures~\cite{S_A_H_WPT_2015}.}}}  In the figures, Monte Carlo simulations are marked with `$\circ$'.

\subsection{User Association}

\begin{figure}
\centering
\includegraphics[height=3.5in,width=4.5in]{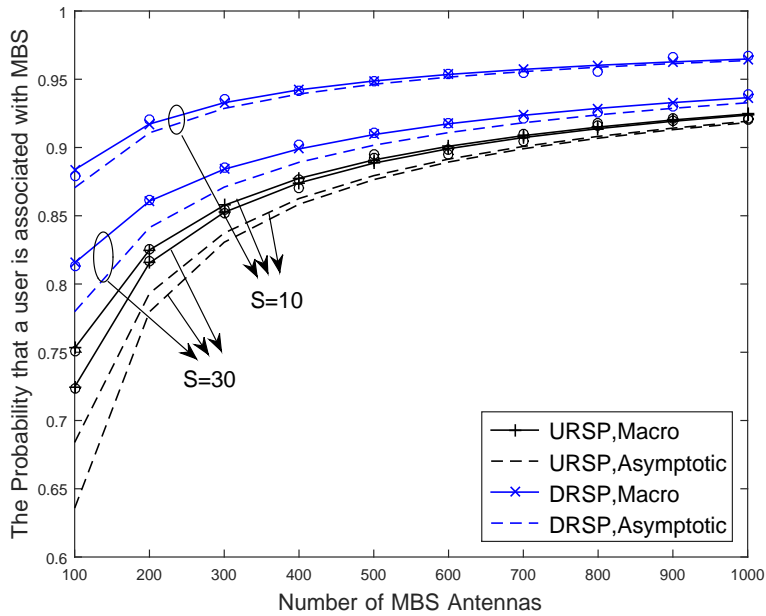}
\caption{Association probability versus the number of antennas for the MBS.}\label{fig:Pro}
\end{figure}

Results in Fig.~\ref{fig:Pro} are provided for the association probability that a user is associated with MBS for various number of MBS antennas. In the results, the path loss exponents were set to $ \alpha_\mathrm{M}= 3.5$, $ \alpha_2  = 4$, and $ \lambda_2=5 \times \lambda_\mathrm{M}$. The solid curves are obtained from \eqref{M_Pr} and \eqref{U_M_Pr} for the DRSP-based and URSP-based user association schemes, respectively, and the dash curves are obtained from the corresponding \eqref{M_Pr_asym} and \eqref{U_M_Pr_asym}, respectively. As we see, our asymptotic expressions can well approximate the exact ones. Also, compared to the URSP-based user association, users are more likely to be served in the macrocells by using DRSP-based user association. The reason is that for DRSP-based user association, MBS provides larger received power. The probability that a user is associated with an MBS increases with the number of MBS antennas, due to the increase of power gain. By increasing $S$, the probability that a user is served by an MBS is reduced due to the decrease of MBS transmit power allocated to each user $\left(\frac{P_\mathrm{M}}{S}\right)$.

\subsection{Downlink Energy Harvesting}
In this subsection, we investigate the energy harvesting performance for different user association schemes presented in Section~\ref{section_WPT}. In the simulations, the block time $T$ is normalized to $1$, while the time allocation factor is $\tau=0.6$, and the path loss exponents are $\alpha_\mathrm{M}= 3$  and $\alpha_2 = 3.5$.

{Fig.~\ref{fig:up-energy_sources} shows the average energy harvested from the directed WPT, isotropic  WPT, and ambient RF for a user associated with MBS based on the DRSP-based user association. The PBS transmit power is $ P_2= 30$ dBm, the density of PBSs is $ \lambda_2=20 \times \lambda_\mathrm{M}$, and $S=20$. We observe that compared to isotropic  WPT and ambient RF,  the directed WPT plays a dominate role in harvesting energy. The average energy harvested from the directed WPT increases with the number of antennas, due to more power gains. The amount of harvested energy from the ambient RF is nearly unaltered when increasing the MBS antennas. However, the average energy harvested from the isotropic WPT slightly decreases with MBS antennas. The reason is that the coverage of the macrocell is expanded by adding more MBS antennas, and the distance between a user and its associated MBS becomes larger on average, which has an adverse effect on the isotropic WPT.}

\begin{figure}
  \centering
\includegraphics[height=3.5in,width=4.5in]{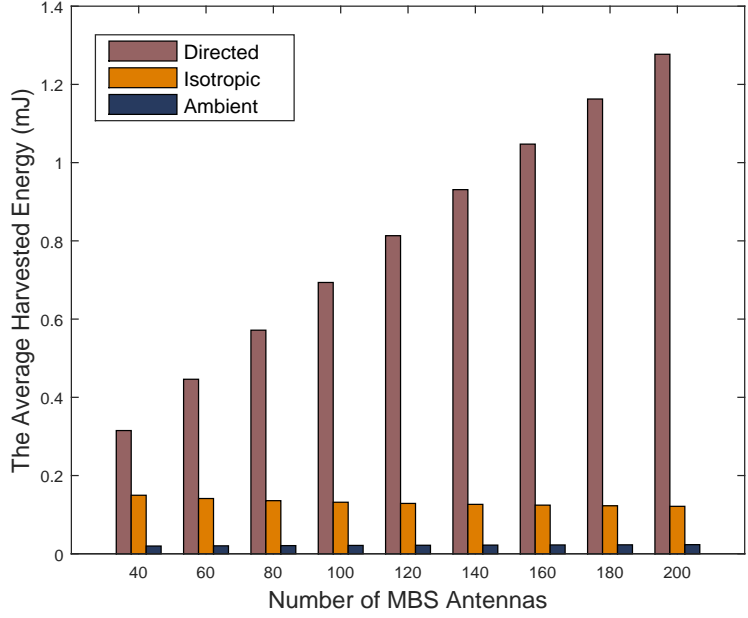}
   \caption{The average harvested energy  against the number of antennas.}\label{fig:up-energy_sources}
\end{figure}

\begin{figure}
  \centering
\includegraphics[height=3.5in,width=4.5in]{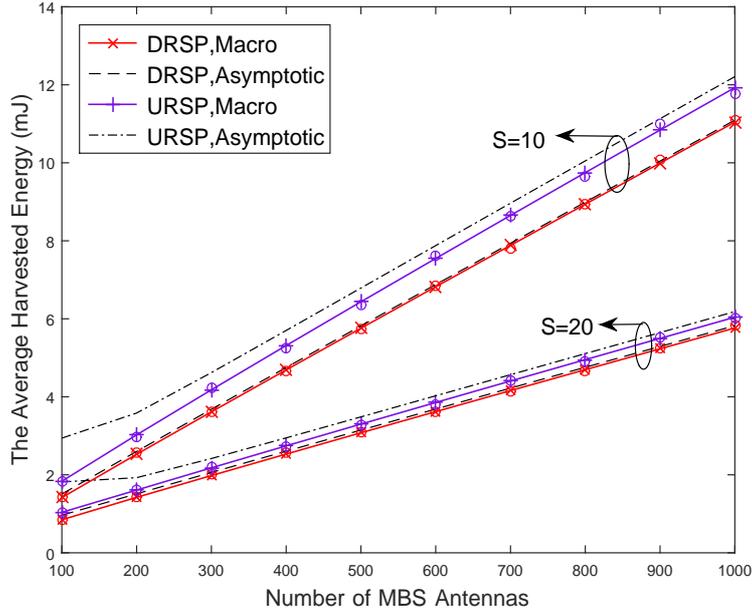}
   \caption{The average harvested energy against the number of antennas for the MBS.}\label{fig:density_mbs}
\end{figure}

Fig.~\ref{fig:density_mbs} shows the average harvested energy of a user associated with the MBS versus the number of MBS antennas. The PBS transmit power is $ P_2= 30{\rm dBm}$ and the density of PBSs is $ \lambda_2=20 \times \lambda_\mathrm{M}$. The solid curves are obtained from \eqref{DRSP_AverEH_M} and \eqref{URSP_AverEH_M}, while the dash curves are obtained from \eqref{DRSP_AverEH_M_asymptotic} and Corollary 4. We see that the asymptotic expressions can well predict the exact ones. The average harvested energy increases with the number of MBS antennas, but decreases with the number of users served by one MBS. {This is because the power gain obtained by the user increases with the number of antennas, but the directed power transfer time allocated to each user decreases with the number of users served by the MBS.} In addition, by URSP-based user association, user in the macrocell harvests more energy than in the case of the DRSP-based user association. { The reason is that with DRSP-based user association, more users with low received power are loaded to the macrocells with increasing number of the MBS antennas.}

\begin{figure}
  \centering
\includegraphics[height=3.5in,width=4.5in]{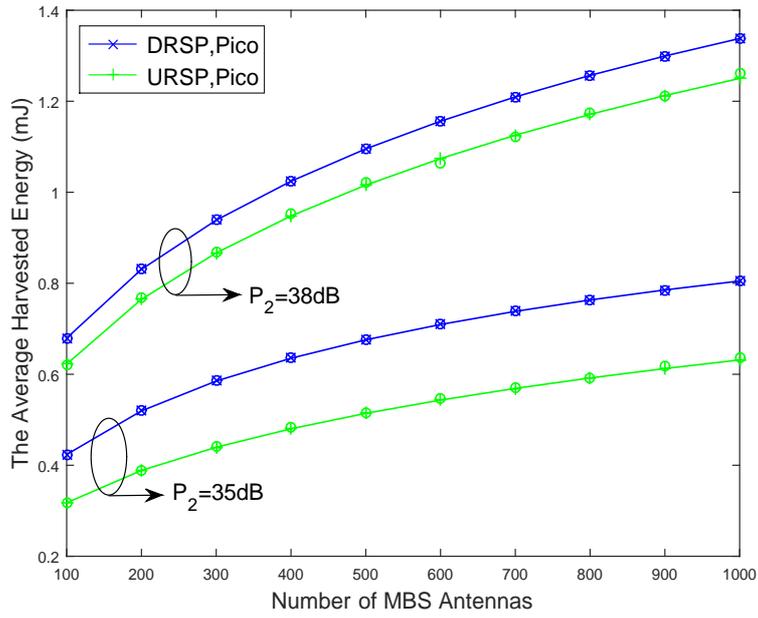}
   \caption{The average harvested energy against the number of antennas for the PBS.}\label{fig:density_pbs}
\end{figure}

Fig.~\ref{fig:density_pbs} shows the average harvested energy of a user associated with the PBS versus the number of MBS antennas.  Here we set $ \lambda_2=20 \times \lambda_\mathrm{M}$ and $ S =5$.  The solid curves are obtained from \eqref{DRSP_AverEH_k} and \eqref{URSP_AverEH_k}. We observe that the harvested energy increases with the number of MBS antennas, due to the fact that users with higher received power are connected to the picocells. Evidently, increasing the PBS transmit power brings an increase on the harvested energy. Moreover, the DRSP based user association outperforms the URSP-based one, since users loaded to the picocells have higher received power through DRSP based user association.

\begin{figure}
  \centering
\includegraphics[height=3.5in,width=4.5in]{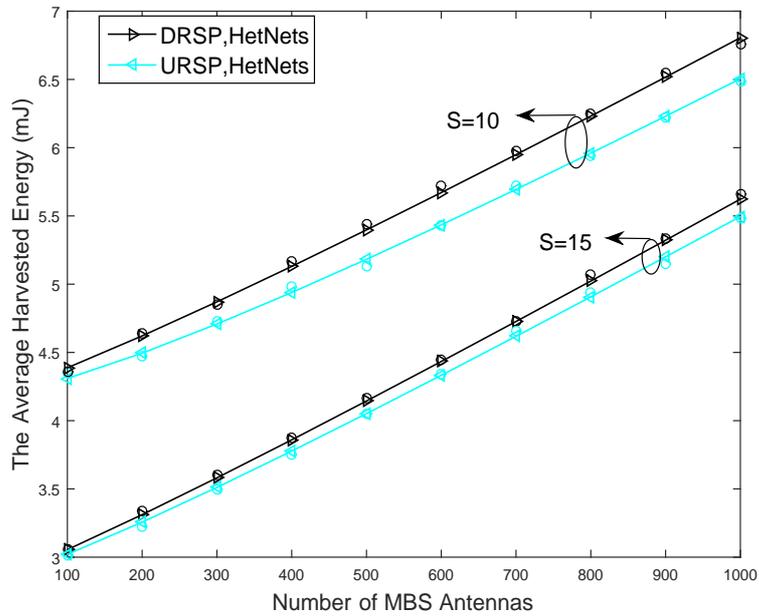}
   \caption{The average harvested energy against the number of antennas in the massive MIMO HetNet.}\label{fig:density_hcn}
\end{figure}

Fig.~\ref{fig:density_hcn} provides the results for the average harvested energy of a user in the massive MIMO HetNet. Same as before, the solid curves are obtained from \eqref{downlink_hcn_energy_drsp} and \eqref{downlink_hcn_energy_ursp}. It is observed that overall, DRSP-based user association harvests more energy than the URSP-based method, since DRSP-based user association seeks to maximize the received power for a user in the HetNet. {In addition, serving more users in the macrocells decreases the harvested energy due to the shorter  directed power transfer time allocated to each user.}

\begin{figure}
\centering
\includegraphics[height=3.5in,width=4.5in]{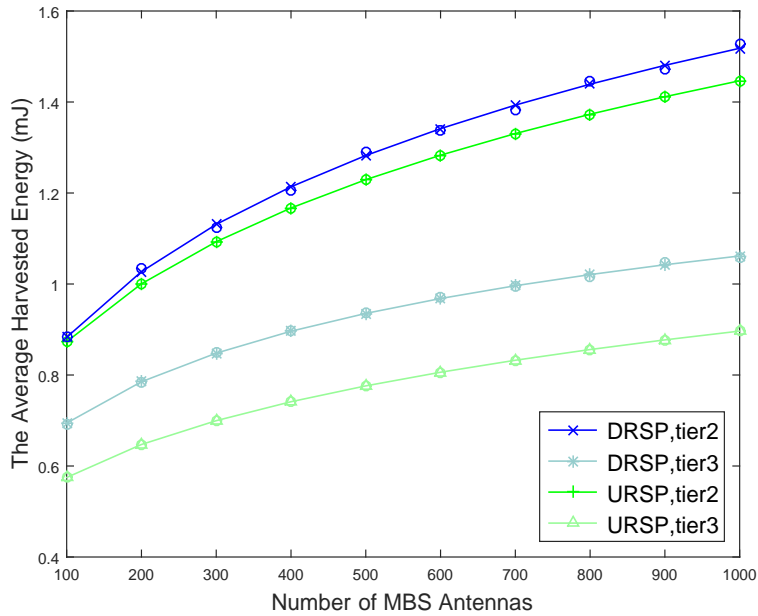}
  \caption{The average harvested energy  against the number of antennas in a three-tier massive MIMO HetNet.}
\label{fig:tauChange1}
\end{figure}

{Fig~\ref{fig:tauChange1} shows the average harvested energy of a user in a three-tier massive MIMO HetNet. In the second and third tier, the densities of BSs are $\lambda _2=20 \lambda_\mathrm{M}$ and $\lambda _3=30 \lambda_\mathrm{M}$, and the BS transmit power are  $P_2=38$ dBm, $P_3=35$ dBm, respectively. We find that compared to the results in Fig. 4, adding another tier can increase the harvested energy of other tiers, because the distances between the BSs and users are shortened. In addition, when adding the number of MBS antennas, the average harvested energy of a user in the second and third tier increases due to the fact that users with low received power are offloaded to macrocells. }

\subsection{Average Uplink Achievable Rate}
In this section, we evaluate the average achievable rate in the uplink, as presented in Section \ref{Section_uplink_rate}. In the simulations, the time allocation factor is $\tau=0.3$, and the path loss exponents are $ \alpha_\mathrm{M}= 2.8$  and $ \alpha_2  = 2.5$, $ P_2= 30{\rm dBm}$ and  $ S= 10$.

\begin{figure}
  \centering
\includegraphics[height=3.5in,width=4.5in]{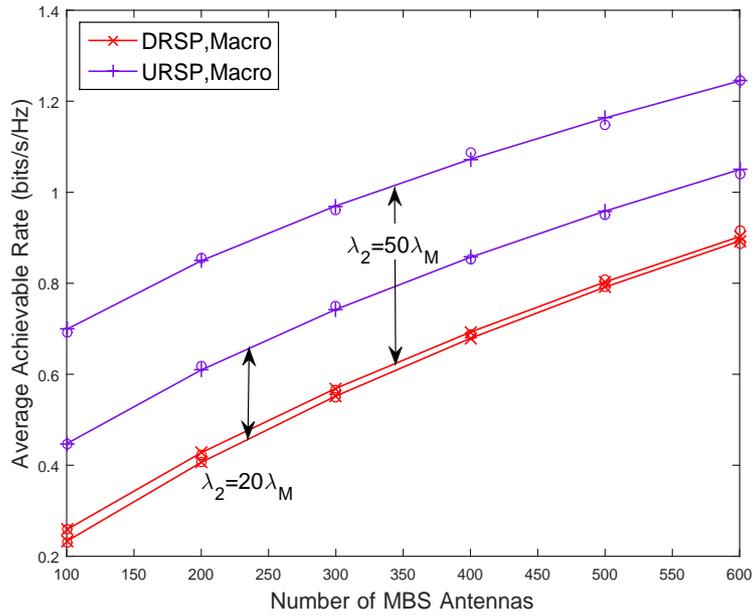}
   \caption{The average uplink achievable rate against the number of antennas for the MBS.}\label{fig:UPLINK_mbs_den}
\end{figure}
\begin{figure}
  \centering
\includegraphics[height=3.5in,width=4.5in]{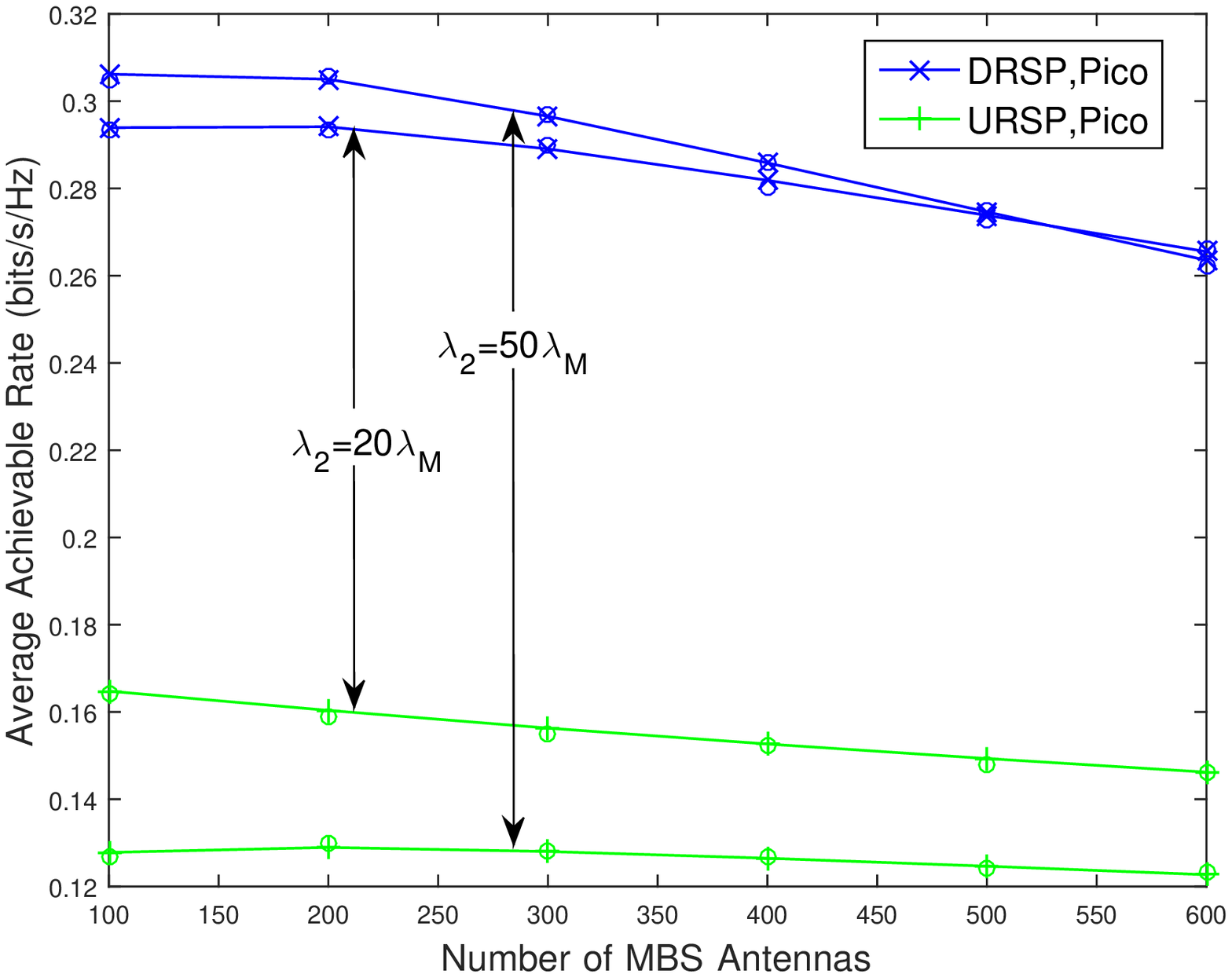}
   \caption{{The average uplink achievable rate against the number of MBS antennas.}}\label{fig:UPLINK_pbs_den}
\end{figure}

\begin{figure}
  \centering
\includegraphics[height=3.5in,width=4.5in]{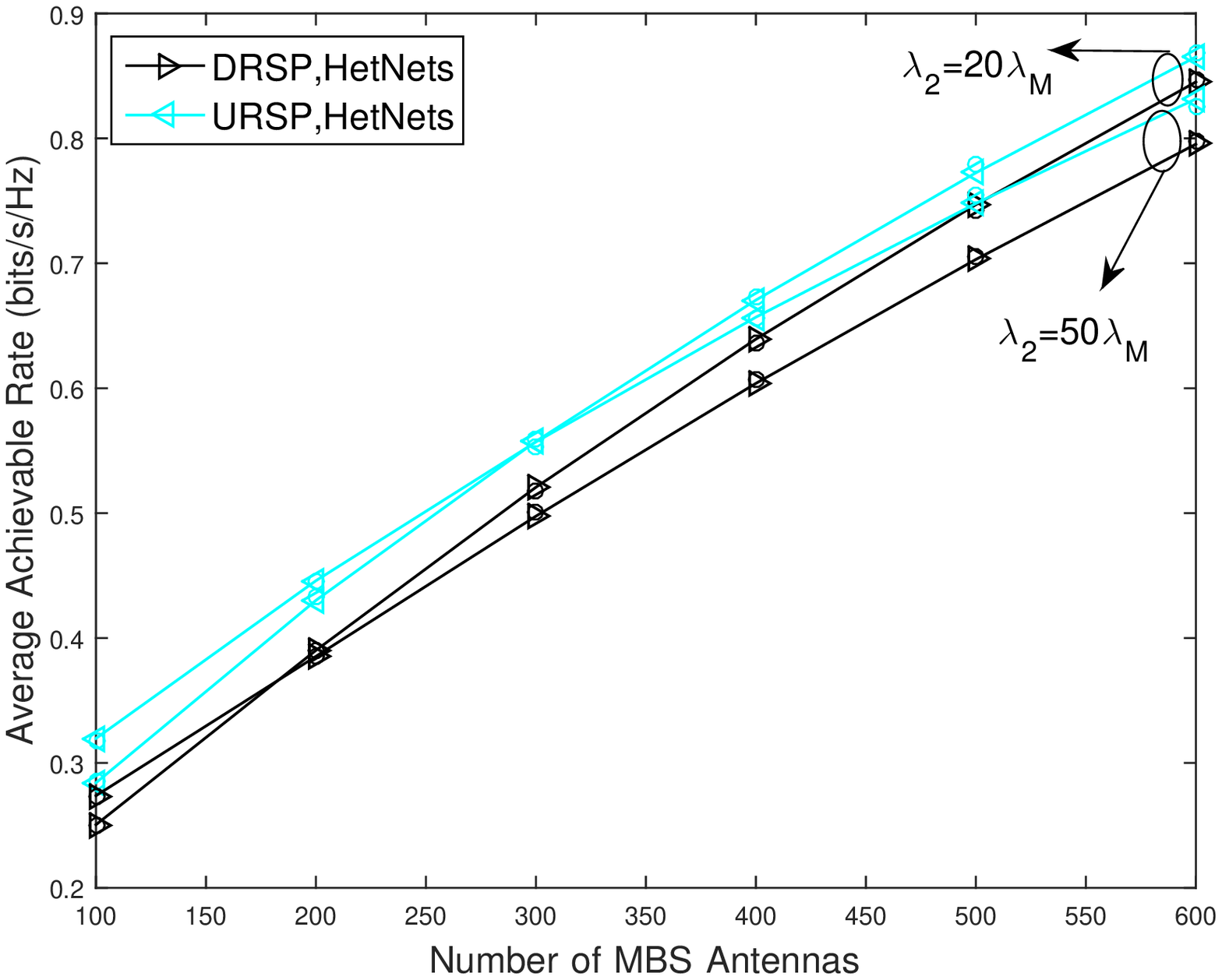}
\caption{The average  uplink achievable rate against the number of antennas in the massive MIMO HetNet.}\label{fig:UPLINK_hcn_den}
\end{figure}
\begin{figure}
  \centering
\includegraphics[height=3.5in,width=4.5in]{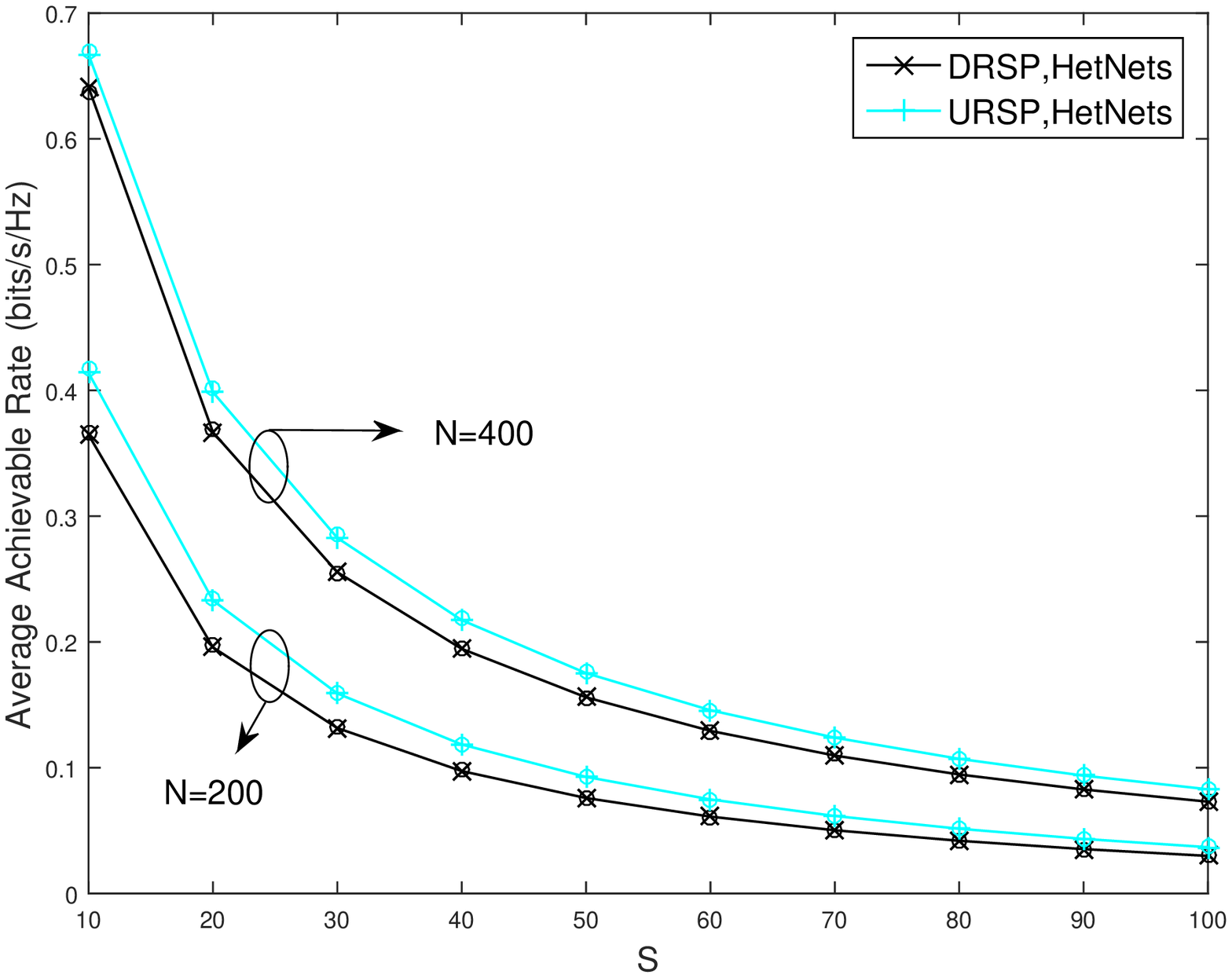}
   \caption{The average  uplink achievable rate against the number of users in the massive MIMO HetNet.}\label{fig:up-hcn-s2}
\end{figure}

Fig.~\ref{fig:UPLINK_mbs_den} shows the average uplink achievable rate of a user associated with the MBS versus the number of MBS antennas. The solid  curves are obtained from \eqref{aver_up_rate_DRSP} and its URSP-based counterpart. We observe that the average achievable rate increases with the number of MBS antennas, due to the increase of the power gain. For URSP-based user association, the average achievable rate also significantly increases with the density of PBSs. The reason is that when the PBSs become more dense, the distance between the user and the PBS is shorter and more users are associated with the PBS, and users with higher received power can be associated with the MBS. However, denser PBSs do not imply a bigger impact on the DRSP-based user association.

 {Fig.~\ref{fig:UPLINK_pbs_den} shows the average uplink achievable rate of a user associated with the PBS versus the number of MBS antennas. The solid  curves are obtained from \eqref{aver_rate_uplink} and its URSP-based counterpart. It is seen that the average achievable rate decreases with increasing the number of MBS antennas. The reason is that users in the macrocells harvest more energy and have higher transmit power, resulting in more severe interference to the uplink in the picocells. Different from the performance behavior  in the macrocells, DRSP-based user association actually outperforms the URDP-based strategy in the picocells.} {In addition, it is indicated from Figs. 7 and 8 that when the PBSs are dense and the number of MBS antennas is not very large, the uplink achievable rate in the picocell can be larger than that in the macrocell under DRSP-based user association.}

Fig.~\ref{fig:UPLINK_hcn_den} demonstrates the results for the average uplink achievable rate in the HetNet. The solid curves are obtained from \eqref{HCN_uplink_rate_drsp} and \eqref{HCN_uplink_rate_Ursp}. Results illustrate that the average rate increases with the number of MBS antennas. Nevertheless, without interference mitigation in the uplink, the deployment of more PBSs deteriorates the uplink performance, since more users are served and more uplink interference exists in the uplink WIT. More importantly, it is indicated that URSP-based user association can achieve better performance than the DRSP-based method, since it seeks to minimize the uplink path loss.  An interesting phenomenon is observed that there is a crossover point, beyond which deploying more PBSs deteriorates the uplink performance due to more uplink interference, {which indicates that in the massive MIMO HetNets with wireless energy harvesting, it can still be interference-limited in the uplink for the dense small cells case, and uplink interference management is needed.}

Finally, Fig.~\ref{fig:up-hcn-s2} shows the average  uplink achievable rate in the HetNet versus $S$. We see that URSP-based user association scheme outperforms the DRSP-based method, and increasing $S$ decreases the average rate, due to more uplink interference and lower harvested energy as suggested in Fig~\ref{fig:density_hcn}.

\section{Conclusions and Future Work}
{In this paper, we considered WPT and  WIT in the massive MIMO enabled HetNets. A stochastic geometry approach was adopted to model the $K$-tier HetNets where massive MIMO  were employed in the macrocells. By addressing the effect of massive MIMO on user association,  we analyzed two specific user association schemes, namely DRSP based scheme for maximizing the harvested energy and URSP based scheme for minimizing the uplink path loss. Based on these two user association schemes, we derived the expressions for the average harvested energy and average uplink rate, respectively. Our results have shown that the use of massive MIMO significantly increases the harvested energy and uplink  rate. When small cells go dense, it can be interference-limited in the uplink.  Although DRSP based user association  has more harvested energy, URSP based user association can achieve higher average uplink rate.}

{Areas that extend the line of this work include imperfect CSI case, and simultaneous wireless information and power transfer (SWIPT) in the downlink. Also, recalling that we have assumed that the number of active users served in each massive MIMO macrocell is a fixed value,  it would be of interest to evaluate the performance by considering the dynamic case. Moreover, it is shown that uplink interference can be severe for dense small cells, and interference management is needed.}



\section*{Appendix A: A proof of Lemma 1} \label{App:A}
\renewcommand{\theequation}{C.\arabic{equation}}
\setcounter{equation}{0}
Using DRSP-based user association in Section II-A, we first examine the power gain by using the proposed downlink power transfer design. As will be indicated by \eqref{B1} in Appendix B, the downlink received power gain is $G_a^{\mathrm{D}}=\left(N+S-1\right)$, which is different from the conventional massive MIMO networks without energy harvesting, due to the fact that the interference is identified as an RF energy source.

Using the similar approach suggested by \cite[Appendix A]{HS12}, we can then obtain the desired results \eqref{DRSP_PDF_M} and \eqref{Lemma-f-down-S}.

\section*{Appendix B: A proof of Theorem 1}\label{App:B}
\renewcommand{\theequation}{B.\arabic{equation}}
\setcounter{equation}{0}
Based on \eqref{MBS_RE_Power}, given $\left|X_{o,\mathrm{M}}\right|=x$, the average harvested energy for a typical user served by the MBS is written as
\begin{align}\label{B1}
&\widetilde{\mathrm{E}}_{o,\mathrm{M}}^{\mathrm{DRSP}}\left(x\right)\notag\\
&=\mathbb{E}\left\{\mathrm{E}_{o,\mathrm{M}}^1\right\}+\mathbb{E}\left\{\mathrm{E}_{o,\mathrm{M}}^2\right\}+\mathbb{E}\left\{\mathrm{E}_{o,\mathrm{M}}^3\right\}\nonumber\\
&=\eta{{{P_\mathrm{M}}}}\left(\mathbb{E}\{h_o\}+\mathbb{E}\{h'_o\}\left(S-1\right)\right)\frac{\tau T}{S}L\left( \max\left\{x,d\right\}  \right) \nonumber\\ &\quad\quad+\mathbb{E}\left\{\mathrm{E}_{o,\mathrm{M}}^3\right\}\nonumber \\
 &=\eta \left({N+S-1}\right) {{\frac{P_\mathrm{M}}{S}}}\beta\notag\\
 &\quad\quad\times\left({\rm{\mathbf{1}}}\left( x \leq d \right) d^{-\alpha_\mathrm{M}}+{\rm{\mathbf{1}}}\left( x > d \right) x^{-\alpha_\mathrm{M}}\right){\tau T}\nonumber\\
 &\quad\quad+\mathbb{E}\left\{\mathrm{E}_{o,\mathrm{M}}^3\right\},
\end{align}
where $\mathbb{E}\left\{\mathrm{E}_{o,\mathrm{M}}^3\right\}$ denotes the average harvested energy from the ambient RF, and is expressed as
\begin{align}\label{B2}
\mathbb{E}\left\{\mathrm{E}_{o,\mathrm{M}}^3\right\}=\eta \left(\mathbb{E}\{I_{\mathrm{M},1}\}+\mathbb{E}\{I_{\mathrm{S},1}\}\right)\times\tau T.
\end{align}
Here, $\mathbb{E}\{I_{\mathrm{M},1}\}$ is the average power harvested from the intra-tier interference, which is given by
\begin{align}\label{B3}
&\mathbb{E}\{I_{\mathrm{M},1}\}\notag\\
&=\mathbb{E}\left\{\sum\limits_{\ell  \in {\Phi_\mathrm{M}}\setminus\left\{ o \right\}} {{P_\mathrm{M}}{h_\ell }L\left(\max \left\{ {\left| {{X_{\ell ,\mathrm{M}}}} \right|},d\right\}\right)}\right\}\nonumber\\
&={P_\mathrm{M}} \mathbb{E} \left\{\sum\limits_{\ell  \in {\Phi_\mathrm{M}}\setminus\left\{ o \right\}} \mathbb{E}\{h_\ell\} L\left(\max \left\{ {\left| {{X_{\ell ,\mathrm{M}}}} \right|},d\right\}\right)\right\}\nonumber\\
&\mathop  = \limits^{\left(a\right)} P_\mathrm{M} \beta 2 \pi {\lambda _{\mathrm{M}}} \left(\int_x^\infty \left(\max \left\{r,d\right\}\right)^{-\alpha_\mathrm{M}} r dr \right)\nonumber\\
&=  P_\mathrm{M} \beta 2 \pi {\lambda _{\mathrm{M}}} \left({\rm{\mathbf{1}}}\left( x \leq d \right) \left({d^{ - {\alpha _{\mathrm{M}}}}} \frac{{({d^2} - {x^2})}}{2} - \frac{{{d^{2 - {\alpha _{\mathrm{M}}}}}}}{{2 - {\alpha _{\mathrm{M}}}}}\right) \right.\nonumber\\
&\quad\quad\quad\quad\quad\quad\quad\quad  \left.-{\rm{\mathbf{1}}}\left( x > d \right) \frac{{{x^{2 - {\alpha _{\mathrm{M}}}}}}}{{2 - {\alpha _{\mathrm{M}}}}} \right),
\end{align}
where $\left(a\right)$ results from $\mathbb{E}\{h_\ell\}=1$ and the Campbell's theorem~\cite{Baccelli2009}. \footnote{{The Campbell's theorem is~\cite{Baccelli2009}: For a Poisson point process $\Phi$ with density $\lambda$, we have $\mathbb{E}\left\{ \sum\limits_{{x_i} \in \Phi } {f\left( {{x_i}} \right)} \right\}=\lambda \int\limits_{{\mathbb{R}^{\dim }}} {\mathbb{E}\left\{ {f\left( x \right)} \right\}dx}$.}} Similarly,  $\mathbb{E}\{I_{\mathrm{S},1}\}$ is the average power harvested from the inter-tier interference, which is given by
\begin{align}\label{B4}
&\mathbb{E}\{I_{\mathrm{S},1}\}\notag\\
&=\mathbb{E}\left\{\sum\limits_{i = 2}^K {\sum\limits_{j \in {\Phi _i}} {{P_i}{h_j}L\left( \max \left\{ {\left| {{X_{j,i}}} \right|},d\right\}\right)} } \right\} \notag\\
&=\sum\limits_{i = 2}^K P_i \beta 2 \pi {\lambda _{i}} \left(\int_{{\hat r}_{{\mathrm{M}}\mathrm{S}}{x^{{\alpha _{\mathrm{M}}}/{\alpha _i}}}}^\infty \left(\max \left\{r,d\right\}\right)^{-\alpha_i} r dr \right)\nonumber\\
&= \sum\limits_{i = 2}^K P_i \beta 2 \pi {\lambda _{i}}\Bigg[{\rm{\mathbf{1}}}\left( x \leq d_o \right)\notag\\
&\quad\quad\quad\quad\quad\quad~\times\left( {d^{- {\alpha _i}}}\frac{{\left({d^2} - {{\hat r}_{{\mathrm{M}}\mathrm{S}}}^2 {x^{\frac{2{\alpha _{\mathrm{M}}}}{\alpha _i}}}\right)}}{2} - \frac{{{d^{2 - {\alpha _i}}}}}{{2 - {\alpha _i}}} \right)\nonumber\\
&\quad\quad\quad\quad\quad\quad\left.-{\rm{\mathbf{1}}}\left( x > d_o \right) \frac{{{\hat r}_{{\mathrm{M}}\mathrm{S}}}^{\left(2-\alpha _i\right)} {x^{\frac{{\alpha _{\mathrm{M}}}(2-\alpha _i) }{\alpha _i}}}}{2-\alpha _i} \right],
\end{align}
in which $d_o=\left({{\hat r}_{{\mathrm{M}}\mathrm{S}}}\right)^{-\frac{\alpha _i}{\alpha_{\mathrm{M}}}}d^{{\alpha _i}/{\alpha _{\mathrm{M}}}}$. By substituting \eqref{B3} and \eqref{B4} into \eqref{B1}, we then obtain \eqref{DRSP_con_AverEH_M}.

We next derive the average harvested energy for a typical user served by the SBS in the $k$-th tier under a given distance
$ \left|X_{o,k}\right|=y$, which is given by
\begin{align}\label{B5}
&\widetilde{\mathrm{E}}_{o,k}^{\mathrm{DRSP}}\left(y\right)\notag\\
&=\mathbb{E}\left\{\mathrm{E}_{o,k}^1\right\} + \mathbb{E}\left\{\mathrm{E}_{o,k}^2\right\}\nonumber \\
&=\eta{{{P_k}}}L\left( \max\left\{y,d\right\}\right)\times\tau T+\eta \left(\mathbb{E} \left\{I_{\mathrm{M},k}\right\}+\mathbb{E} \left\{I_{\mathrm{S},k}\right\}\right)\times\tau T,
\end{align}
where $\mathbb{E} \left\{I_{\mathrm{M},k}\right\}$ is calculated as
\begin{align}\label{B6}
&\mathbb{E} \left\{I_{\mathrm{M},k}\right\}\notag\\
&=\mathbb{E} \left\{\sum\limits_{\ell  \in {\Phi_\mathrm{M}}} {{P_\mathrm{M}}{g_\ell }L\left(\max \left\{ {\left| {{X_{\ell,\mathrm{M}}}} \right|},d\right\}\right)}\right\}\nonumber\\
&=P_\mathrm{M} \beta 2 \pi {\lambda _{\mathrm{M}}} \left(\int_{{{\hat r}_{\mathrm{SM}}}{{y^{{\alpha _k}/{\alpha _{\mathrm{M}}}}}}}^\infty \left(\max \left\{r,d\right\}\right)^{-\alpha_\mathrm{M}} r dr \right)\nonumber\\
&=P_\mathrm{M} \beta 2 \pi {\lambda _{\mathrm{M}}} \Bigg[{\rm{\mathbf{1}}}\left( y \leq d_1 \right)\notag\\
&\quad\quad\quad\quad\quad\times\left({d^{ - {\alpha _{\rm M}}}}{\frac{{\left({d^2} - {{\hat r}_{\mathrm{SM}}^2}{{y^{\frac{2\alpha _k}{\alpha _{\mathrm{M}}}}}}\right)}}{2} - \frac{{{d^{2 - {\alpha _{\rm M}}}}}}{{2 - {\alpha _{\rm M}}}}}\right)\nonumber\\
&\quad\quad\quad\quad\quad\left.- {\rm{\mathbf{1}}}\left( y > d_1 \right)  \frac{{\hat r}_{\mathrm{SM}}^{2-\alpha _{\rm M}} y^{\frac{\alpha _k\left(2-\alpha _{\rm M}\right)}{\alpha _{\mathrm{M}}}} }{{2 - {\alpha _{\rm M}}}} \right],
\end{align}
where $d_1={\left({\hat r}_{\mathrm{SM}}\right)^{\frac{-\alpha _{\rm M}}{\alpha _k}}}{d^{{\alpha _{\rm M}}/{\alpha _k}}}$, and  $\mathbb{E} \left\{I_{\mathrm{S},k}\right\}$ is given by
\begin{align}\label{B7}
&\mathbb{E} \left\{I_{\mathrm{S},k}\right\}\notag\\
&=\mathbb{E} \left\{\sum\limits_{i = 2}^K {\sum\limits_{j \in {\Phi _i}\setminus \left\{o\right\}} {{P_i}{g_{j,i}}L\left( \max \left\{ {\left| {{X_{j,i}}} \right|},d\right\}\right)} }\right\} \nonumber\\
&=\sum\limits_{i = 2}^K  \beta 2 \pi {\lambda_i} \int_{{{\hat r}_{\mathrm{SS}}}{y^{\frac{\alpha _k}{\alpha _i}}}}^\infty \left(\max \left\{r,d\right\}\right)^{-\alpha_i} r dr \nonumber\\
&=\sum\limits_{i = 2}^K  \beta 2 \pi {\lambda_i} \Bigg[{\rm{\mathbf{1}}}\left( y \leq d_2 \right)\notag\\
& \quad\quad\quad\quad\quad\times\left({{d^{ - {\alpha _i}}} \frac{{\left({d^2} - {{\hat r}_{\mathrm{SS}}^2}y^{\frac{2 \alpha _k}{\alpha _i}} \right)}}{2} - \frac{{{d^{2 - {\alpha _i}}}}}{{2 - {\alpha _i}}}}  \right)\nonumber\\
&\quad\quad\quad\quad\quad\left.- {\rm{\mathbf{1}}}\left( y > d_2 \right) \frac{{{\hat r}_{\mathrm{SS}}}^{2 - {\alpha _i}}y^{\frac{ \alpha _k \left(2 - {\alpha _i}\right)}{\alpha _i}}  }{{2 - {\alpha _i}}}  \right],
\end{align}
where  $d_2={\left({{\hat r}_{\mathrm{SS}}}\right)^{\frac{-\alpha _i}{\alpha _k}}}{d^{{\alpha _i}/{\alpha _{k}}}}$. By plugging \eqref{B6} and \eqref{B7} into \eqref{B5}, we obtain the desired result in \eqref{DRSP_con_AverEH_k}.

\section*{Appendix C: A proof of Corollary 3}\label{App:C}
\renewcommand{\theequation}{C.\arabic{equation}}
\setcounter{equation}{0}
According to \eqref{DRSP_con_AverEH_M} and \eqref{DRSP_AverEH_M}, we first are required to derive the following asymptotic expressions:
\begin{subequations}
\begin{align}\label{C1}
{\Xi _1}\left( x \right)&=\int_{\rm{0}}^x {f_{\left| {{X_{o,{\rm{M}}}}} \right|}^{{\rm{DRSP}}}(r)dr} ,\\
{\Xi _2}\left( {a{\rm{,}}b} \right)&=\int_a^\infty  {{x^b}f_{\left| {{X_{o,{\rm{M}}}}} \right|}^{{\rm{DRSP}}}(x)dx}, \\
{\Xi _3}\left( {c{\rm{,}}d} \right)&=\int_0^c {{x^d}f_{\left| {{X_{o,{\rm{M}}}}} \right|}^{{\rm{DRSP}}}(x)dx}.
\end{align}
\end{subequations}
By using the Taylor series expansion truncated to the first order  as $N \rightarrow \infty$, (C.1a) is asymptotically computed as
\begin{multline}\label{C2}
{\Xi _1}\left( x \right)= \frac{2\pi {\lambda _{\rm{M}}}}{{\Psi _{{{\rm{M}}_\infty }}^{{\rm{DRSP}}}}}\left[ {\int_0^x {r\exp \left( { - \pi {\lambda _{\rm{M}}}{r^2}} \right)dr} } \right.\\
\left. - \pi \sum\limits_{i = 2}^K {{\lambda _i}} {\hat r_{{\rm{MS}}}^2\int_0^x {{r^{1 + \frac{2{\alpha _{\rm{M}}}}{{\alpha _i}}}}\exp \left( { - \pi {\lambda _{\rm{M}}}{r^2}} \right)dr} } \right].
\end{multline}
It is noted that the asymptotic expression for the probability of a typical user that is associated with the MBS has been derived in \eqref{M_Pr_asym}. Therefore, we can directly apply the result in \eqref{C2}. After some mathematical manipulations, we obtain \eqref{asym_CDF_exp_M}.  Similarly, the asymptotic expressions for (C.1b) and (C.1c) are correspondingly derived as \eqref{E_a_b_function} and \eqref{E_a_b2_function}. Substituting \eqref{asym_CDF_exp_M}--\eqref{E_a_b2_function} into \eqref{DRSP_AverEH_M}, we obtain the desired result in \eqref{DRSP_AverEH_M_asymptotic}.

\section*{Appendix D: A proof of Theorem 2}\label{App:D}
\renewcommand{\theequation}{D.\arabic{equation}}
\setcounter{equation}{0}
The exact average achievable rate is written as
\begin{equation}
R=\frac{\left(1-\tau\right)T}{T}\mathbb{E}\left\{\log_2\left(1+{\rm SINR}\right)\right\}.
\end{equation}
Now, using Jensen's inequality, we can obtain the lower bound for the conditional average uplink achievable rate between a typical user and its serving MBS as
\begin{align}\label{D1}
 R_{\mathrm{DRSP},\mathrm{M}}^\mathrm{low}\left(x\right)=\left(1-\tau\right)\log_2
 \left(1+\frac{1}{\mathbb{\mathrm{E}}\left\{\mathrm{SINR}_\mathrm{M}^{-1}\right\}}\right).
\end{align}
Based on \eqref{SINR_Macro}, ${\mathbb{E}\left\{\mathrm{SINR}_\mathrm{M}^{-1}\right\}}$ is calculated as
\begin{align}\label{D2}
&{\mathbb{E}\left\{\mathrm{SINR}_\mathrm{M}^{-1}\right\}}\notag\\
&= \mathbb{E}\left\{\frac{{I_{u,\mathrm{M}}} + {I_{u,\mathrm{S}}} + {\delta ^2}}{{P_{{u_\mathrm{M}}}^{\mathrm{DRSP}}}{h_{o,\mathrm{M}}}L\left( \max\left\{x,d\right\} \right)}\right\}\nonumber\\
&\overset{(a)}{\approx}\left({{P_{{u_\mathrm{M}}}^{\mathrm{DRSP}}}{\left(N-S+1\right)}L\left(\max\left\{x,d\right\} \right)}\right)^{-1}  \nonumber\\
&\quad\quad\times\left(\mathbb{E}\left\{I_{u,\mathrm{M}}\right\}+\mathbb{E}\left\{I_{u,\mathrm{S}}\right\}+\delta^2\right),
\end{align}
where (a) is obtained by using the law of large numbers, i.e., $h_{o,\mathrm{M}} \approx N-S+1 $ as $N$ becomes large. Using the Campbell's theorem~\cite{Baccelli2009}, we next derive $\mathbb{E}\left\{I_{u,\mathrm{M}}\right\}$ as
\begin{align}\label{D3}
&\mathbb{E}\left\{I_{u,\mathrm{M}}\right\}\notag\\
&= \mathbb{E}\left\{\sum\limits_{i \in {{\widetilde {\mathcal{U}}}_\mathrm{M}} \setminus \left\{ o \right\}} {{P_{{u_\mathrm{M}}}^{\mathrm{DRSP}}}{h_i}L\left( {\max \left\{ {\left| {{X_i}} \right|,d} \right\}} \right)}\right\}\nonumber\\
&={P_{{u_\mathrm{M}}}^{\mathrm{DRSP}}} \beta 2 \pi (S\lambda_\mathrm{M}) \left(\int_0^d {{d^{ - {\alpha_\mathrm{M}}}}r} dr + \int_d^\infty  {{r^{ - {\alpha_\mathrm{M}}}}} r dr\right)\nonumber\\
&={P_{{u_\mathrm{M}}}^{\mathrm{DRSP}}} \beta 2 \pi (S\lambda_\mathrm{M})\left(\frac{{{d^{2 - {\alpha _{\rm M}}}}}}{2} + \frac{{{d^{2 - {\alpha _{\rm M}}}}}}{{{\alpha _{\rm M}}-2 }}\right).
\end{align}
Likewise, $\mathbb{E}\left\{I_{u,\mathrm{S}}\right\}$ is derived as
\begin{align}\label{D4}
&\mathbb{E}\left\{I_{u,\mathrm{S}}\right\}\notag\\
&=\mathbb{E}\left\{\sum\limits_{i = 2}^K {\sum\limits_{j \in {{\widetilde {\mathcal{U}}}_i}} {{P_{{u_i}}^{\mathrm{DRSP}}}{h_j}L\left( {\max \left\{ {\left| {{X_j}} \right|,d} \right\}} \right)} }\right\}\nonumber\\
&= \sum\limits_{i = 2}^K {P_{{u_i}}^{\mathrm{DRSP}}}\beta 2 \pi \lambda_i \left(\frac{{{d^{2 - {\alpha _{\rm M}}}}}}{2} + \frac{{{d^{2 - {\alpha_{\rm M}}}}}}{{{\alpha _{\rm M}}-2 }}\right).
\end{align}
Substituting  \eqref{D2}--\eqref{D4} into \eqref{D1}, we obtain \eqref{theo_2}.

\section*{Appendix E: A proof of Theorem 3}\label{App:E}
\renewcommand{\theequation}{E.\arabic{equation}}
\setcounter{equation}{0}
Given a distance $\left|X_{o,k}\right|=y$, the conditional average uplink achievable rate for a typical user served by the SBS in the $k$-th tier is expressed as
\begin{align}\label{E1}
R_{\mathrm{DRSP},k}\left(y\right)&=\frac{\left(1-\tau\right)T}{T}\mathbb{E}\left\{\log_2\left(1+\mathrm{SINR}_k\right)\right\} \nonumber\\
&=\frac{\left(1-\tau\right)}{\ln 2}\int_0^\infty \frac{{\bar{F}}_{\mathrm{SINR}}\left(x\right)}{1+x} dx,
\end{align}
where ${\bar{F}}_{\mathrm{SINR}_k}\left(x\right)$ is the CCDF of the received SINR, denoted by $\mathrm{SINR}_k$, and is given by
\begin{align}\label{E2}
&{\bar{F}}_{\mathrm{SINR}}\left(x\right)\notag\\
&=\Pr\left(\mathrm{SINR}_k>x\right)\nonumber\\
&=\Pr\left(\frac{{{P_{u_k}^{\mathrm{DRSP}}}{g_{o,k}}L\left( {y,d} \right)}}{{{I_{u,\mathrm{M}}} + {I_{u,\mathrm{S}}} + {\delta ^2}}}>x\right)\nonumber\\
&=e^{-\frac{x {\delta ^2}}{P_{u_k}^{\mathrm{DRSP}}\Delta_2\left(y\right)}}\mathbb{E}\left\{e^{-\frac{x {I_{u,\mathrm{M}}}}{P_{u_k}^{\mathrm{DRSP}}\Delta_2\left(y\right)}}\right\} \mathbb{E}\left\{e^{-\frac{x {I_{u,\mathrm{S}}}}{P_{u_k}^{\mathrm{DRSP}}\Delta_2\left(y\right)}}\right\}\nonumber\\
&=e^{-\frac{x {\delta ^2}}{P_{u_k}^{\mathrm{DRSP}}\Delta_2\left(y\right)}}\times\notag\\
&\quad\quad\quad{{\cal{L}}_{{I_{u,\mathrm{M}}}}}\left(\frac{x }{P_{u_k}^{\mathrm{DRSP}}\Delta_2\left(y\right)}\right){{\cal{L}}_{{I_{u,\mathrm{S}}}}}\left(\frac{x }{P_{u_k}^{\mathrm{DRSP}}\Delta_2\left(y\right)}\right),
\end{align}
where $\Delta_2\left(y\right)=L\left( \max\left\{y,d\right\} \right)$, ${{\cal{L}}_{{I_{u,\mathrm{M}}}}}\left(\cdot\right)$ and ${{\cal{L}}_{{I_{u,\mathrm{S}}}}}\left(\cdot\right)$ are the Laplace transforms of the PDFs of $I_{u,\mathrm{M}}$ and $I_{u,\mathrm{S}}$, respectively. {Considering the fact that users are densely served in the massive MIMO HetNets, the minimum distance between the typical BS and the interfering users is small, the Laplace transform of the PDF of $I_{u,\mathrm{M}}$ can be approximately derived as~\cite{XinqinLin_2014}}
\begin{align}\label{E3}
&{{\cal{L}}_{{I_{u,\mathrm{M}}}}}\left(s\right)\notag\\
&=\mathbb{E}\left\{\exp\left(-s\sum\limits_{i \in {{\widetilde {\mathcal{U}}}_\mathrm{M}} } {{P_{u_\mathrm{M}}^{\mathrm{DRSP}}}{g_i}L\left( {\max \left\{ {\left| {{X_i}} \right|,d} \right\}} \right)}\right)\right\}\nonumber\\
&\mathop  \approx \limits^{\left(a\right)} \exp\left(-2\pi(S\lambda_\mathrm{M})\int_0^\infty \frac{s P_{u_\mathrm{M}}^{\mathrm{DRSP}}L\left( {\max \left\{ {r,d} \right\}} \right)}{1+s P_{u_\mathrm{M}}^{\mathrm{DRSP}}L\left( {\max \left\{ {r,d} \right\}} \right)} r dr \right)\nonumber\\
&= \exp\left(- \pi (S\lambda_\mathrm{M}) \frac{s P_{u_\mathrm{M}}^{\mathrm{DRSP}}\beta d^{-\alpha_i}}{1+s P_{u_\mathrm{M}}^{\mathrm{DRSP}} \beta d^{-\alpha_i}} d^2\right.\notag\\
&\quad\quad -2\pi(S \lambda_\mathrm{M}) s P_{u_\mathrm{M}}^{\mathrm{DRSP}} \beta\times\nonumber\\
&\quad\quad~~ \left.  \frac{ d^{2-\alpha_i}}{\alpha_i-2} {}_2{F_1}
\left[1,\frac{{{\alpha _i} - 2}}{{{\alpha _i}}};2 - \frac{2}{{{\alpha _i}}}; - s P_{u_\mathrm{M}}^{\mathrm{DRSP}}\beta {{d}^{ - {\alpha _i}}} \right]   \right),
\end{align}
where (a) is obtained by using the generating functional of PPP~\cite{M_Haenggi2013}. Similarly, ${{\cal{L}}_{{I_{u,\mathrm{S}}}}}\left(s\right)$ is given by
\begin{align}\label{E4}
&{{\cal{L}}_{{I_{u,\mathrm{S}}}}}\left(s\right)\notag\\
&\approx \exp\left(-\sum\limits_{i = 2}^K \pi \lambda_i  \frac{s P_{u_i}^{\mathrm{DRSP}}\beta d^{-\alpha_i}}{1+s P_{u_i}^{\mathrm{DRSP}} \beta d^{-\alpha_i}} d^2\right.\notag\\
&\quad\quad -\sum\limits_{i = 2}^K 2\pi\lambda_i s P_{u_i}^{\mathrm{DRSP}} \beta \frac{ d^{2-\alpha_i}}{\alpha_i-2}\times \nonumber\\
&\quad\quad~~\left.  {}_2{F_1}
\left[1,\frac{{{\alpha _i} - 2}}{{{\alpha _i}}};2 - \frac{2}{{{\alpha _i}}}; - s P_{u_i}^{\mathrm{DRSP}}\beta {{d}^{ - {\alpha _i}}} \right]   \right).
\end{align}
Substituting \eqref{E3} and \eqref{E4} into \eqref{E2}, we get \eqref{CCDF_tier_k}.


\end{document}